\documentclass[]{interact}

\usepackage{epstopdf}
\usepackage[caption=false]{subfig}
\usepackage{float}
\usepackage[numbers,sort&compress]{natbib}
\bibpunct[, ]{[}{]}{,}{n}{,}{,}

\theoremstyle{plain}
\newtheorem{theorem}{Theorem}[section]

\newtheorem{proposition}[theorem]{Proposition}

\theoremstyle{definition}

\theoremstyle{remark}

\numberwithin{equation}{section}

\usepackage{anyfontsize}

\begin{document}

\articletype{}

\title{Heterogenous Macro-Finance Model: A Mean-field Game Approach}

\author{
\name{Hoang Huy Vu\textsuperscript{a}\thanks{CONTACT Hoang Huy Vu. Email: hoangvu4797@gmail.com} and Tomoyuki Ichiba\textsuperscript{a}\thanks{CONTACT Tomoyuki ICHIBA. Email: ichiba@pstat.ucsb.edu}}
\affil{\textsuperscript{a}Department of Statistics and Applied Probability, University of California, Santa Barbara, California 93106, United States 
}
}
\maketitle

\begin{abstract}
We investigate the full dynamics of capital allocation and wealth distribution of heterogeneous agents in a frictional economy during booms and busts using tools from mean-field games. Two groups in our models, namely the expert and the household, are interconnected within and between their classes through the law of capital processes and are bound by financial constraints. Such a mean-field interaction explains why experts accumulate a lot of capital in the good times and reverse their behavior quickly in the bad times even in the absence of aggregate macro-shocks. When common noises from the market are involved, financial friction amplifies the mean-field effect and leads to capital fire sales by experts. In addition, the implicit interlink between and within heterogeneous groups demonstrates the slow economic recovery and characterizes the deviating and fear-of-missing-out (FOMO) behaviors of households compared to their counterparts. Our model also gives a fairly explicit representation of the equilibrium solution without exploiting complicated numerical approaches.
\end{abstract}


\section{Introduction}
The persistence of significant economic crises, despite decades of study since the Great Depression, underscores the need for a deeper understanding of the underlying mechanisms. Although early work by Fisher (1933), Keynes (1936), and Minsky (1957) linked recessions to financial market failures, these insights remained underexplored, leading to a divergence between macroeconomics and finance. A key recent development is the effort to reunify the fields through continuous-time stochastic models (e.g., Sannikov, 2008). This approach acknowledges the crucial role of the financial sector in shaping the broader economic dynamics. Financial markets, while exhibiting self-correction during stability, become vulnerable during periods of debt accumulation and credit bubbles. This susceptibility arises from the inherent characteristics of the business cycle – persistence and amplification (Bernanke et al., 1999; Kiyotaki \& Moore, 1997). These studies employ macroeconomic models to analyze the dynamics of the system around its equilibrium state.\\

Following the 2008 crisis, research has emphasized the role of intermediary leverage in influencing asset prices (e.g., He \& Krishnamurthy, 2013). This is crucial due to the significant role intermediaries play in transmitting shocks to the real economy. Furthermore, limitations of traditional representative agent models, assuming homogeneous preferences and perfect rationality, are evident in explaining observed improper asset allocation and behavioral deviations during downturns. Consequently, recent research has shifted towards heterogeneous agent frameworks. These models incorporate the dispersion of preferences, risk aversion, and information asymmetry among economic actors. This allows analysis of herding behavior, where individuals mimic peers, potentially leading to positive feedback loops and amplifying market movements. Furthermore, bounded rationality and imperfect information processing can induce panic-driven asset sales or excessive risk taking, further exacerbating crises.\\

This paper investigates the full equilibrium dynamic, not just near the steady state while studying the whole distribution of agent interaction, not just the first moment as in Krusell \& Smith (1998). We focus on the dynamic capital allocation behavior of heterogeneous agents within booms and busts, and aim to explain the observed discrepancies in their choices and the resulting implications for economic recovery. To achieve this, we construct a model of a competitive market populated by two distinct groups of agents: {\it experts}, possessing superior information and potentially higher risk tolerance, and {\it households}, representing the broader population with different levels of information and risk aversion. Our model also incorporates financial frictions that represent real-world impediments to the free flow of capital between agents. We posit that financial frictions play a crucial role in hindering the optimal allocation of capital across different sectors of the economy during periods of economic distress. It can impede the movement of capital to sectors experiencing capital shortages during downturns. This prevents the crucial process of economic recovery, restructuring, and reallocation of resources. In the absence of perfect capital mobility, sectors with high growth potential might be starved of necessary funds, while resources remain trapped in declining sectors. This misallocation of resources slows down the overall economic recovery process. Further complicating the recovery process, financial frictions can exacerbate the issue of excessive debt burdens in a scenario with heterogeneous agents. Experts, facing limited access to credit markets due to these frictions, might struggle to deleverage effectively. This further hinders economic growth, as experts with high debt burdens tend to reduce consumption and investment, dampening overall economic activity. \\

A scenario characterized by perfect financial markets from a single representative agent model would exhibit vastly different dynamics compared to our framework with financial frictions. Perfect markets imply unimpeded capital flow, allowing capital to readily flow to sectors with the greatest need, facilitating a more rapid and efficient recovery process. In addition, households would be able to effectively manage their debt burdens due to their unconstrained access to credit markets. Finally, perfect markets would lead to optimal risk sharing, where risks are efficiently distributed across the economy based on individual risk tolerances, creating a more resilient system. However, such a scenario remains an unrealistic theoretical ideal. By incorporating financial frictions into our model, we capture a more realistic representation of economic behavior and highlight their potentially significant contribution to slow economic recoveries.\\

The heterogeneity of the agents in capital production and risk tolerance plays a critical role in the slow recovery process observed in our model with financial frictions. This heterogeneity can be robustly analyzed and modeled using a recent mathematical innovation – the mean-field game (MFG) (see, e.g., Achdou et al., 2022). The MFG framework utilizes a system of interconnected equations: the Backward Hamilton-Jacobi-Bellman (HJB) equation, which captures the optimal decision-making process of individual agents, considering the equilibrium distribution of other agents' strategies in the system; and the Kolmogorov Forward (KF) equation, which describes the evolution of the distribution of agents' states, preferences, and beliefs over time, taking into account individual decisions and their aggregate impact. This backward-forward MFG system provides a powerful tool for analyzing any heterogeneous agent model characterized by a continuum of agents with heterogeneous beliefs. It offers several key advantages for analyzing complex economic systems with a large number of agents.  First, it allows for tractable analysis by approximating agent interactions through the equilibrium distribution. This is particularly valuable when dealing with a large number of individual actors. Second, the framework identifies the Nash equilibrium in the game, where no individual agent has an incentive to deviate from their chosen strategy given the strategies of others. Finally, the interconnected HJB and KF equations allow for the analysis of how the distribution of beliefs and behaviors evolves dynamically within the system.\\
 
Using the MFG framework, we can construct the Nash closed-loop problem and solve the full dynamic equilibrium explicitly without resorting to computationally expensive numerical methods. The existing literature, which often employs a finite difference scheme on models with a limited number of state variables (e.g., Brunnermeier \& Sannikov, 2016), suffers from limitations in scalability and accessibility. Grid-based approaches become computationally expensive for high-dimensional models, and model calibration often requires extensive assessments and customization of numerical techniques. Although recent proposals utilizing deep learning approaches (Gopalakrishna, 2022; Fan et al., 2023) offer some promise, they remain susceptible to limitations such as data dependency and restricted applicability to non-stationary settings. \\

Our model explicitly incorporates a mean-field interaction between and within expert and household agent groups, modeled through the probability distribution of capital processes. We demonstrate the importance of this interaction by examining asset allocation behavior in two scenarios: with and without aggregate shocks. In the absence of aggregate shocks, distinct behavior patterns emerge between households and experts. During high-growth periods, households strategically reduce their capital holdings while experts accumulate capital. This behavior is rational – households recognize the superior capital management skills of experts during booms and optimally choose to delegate their capital. However, during downturns with negative shocks, experts attempt to minimize their exposure to mitigate liquidity and continuity risks. This motivates households to increase their capital holdings, essentially supporting experts during periods of stress. Once experts recover, households gradually reduce their capital holdings again, returning control to experts. The presence of aggregate shocks, which affect both experts and households, highlights the role of financial frictions. Financial frictions exacerbate the effects of negative shocks, leading to more severe asset fire sales by experts. On the other hand, the mean-field interaction amplifies the cautious behavior of households, which leads them to be less aggressive in absorbing capital from experts during crises. Our model also sheds light on the sluggish nature of economic recovery through the lens of mean-field interaction. Unlike homogeneous agent models in the current macro-finance literature, which often depict rapid recoveries due to less constrained expert behavior, our model captures the need for experts to observe and adapt their behavior based on the actions of others within their group. This deliberation process contributes to slower recovery dynamics. Furthermore, the mean-field interaction also reflects the fear of missing out (FOMO) behavior of households. Instead of waiting for the peak of the expert fire sale, households strategically increase their capital holdings during boom periods.\\

 The papers closest to ours are Brunnermeier and Sannikov (2014); and Gopalakrishna (2022). Brunnermeier and Sannikov (2014) introduced a model with only heterogeneity between agents, featuring a representative expert and a representative household. Their work focuses on wealth distribution and leverages the concept of the marginal value of wealth $\theta_t$ and its reflecting boundary. However, the model they used to describe $\theta_t$ raised concern about the need to add the local time to the law of the process. A detailed discussion of this falls outside the scope of this paper, so we encourage readers to examine their mathematical derivations in their paper. In addition, their model is incapable of describing the slow recovery of the economy as it is supposed to be with the inclusion of financial frictions. Gopalakrishna (2022) attributes the slow recovery in Brunnermeier and Sannikov's model (2014) to their assumption of constant productivity. He argues that by allowing intermediaries to maintain a constant productivity rate throughout the state space and remain operational, a trade-off emerges between the duration of crises and the conditional risk premium.  During economic downturns, asset price inefficiencies arise due to difficulties in capital allocation and asset sales. This creates an opportunity for sophisticated investors with borrowed capital to gain higher risk premiums, potentially accelerating their post-crisis wealth recovery. Gopalakrishna addresses this limitation by introducing a model with stochastic productivity and economic regime-dependent intermediary exit. In contrast to Gopalakrishna's (2022) work, we demonstrate that even under the assumption of constant productivity, slow economic recoveries can be observed by introducing sufficient constraints on expert behavior. One such constraint is the inclusion of mean-field interaction, which captures the dynamic interdependence between experts within the model. That is, slow recovery can arise not only from the specific characteristics of intermediaries but also from the strategic interactions and decision-making processes of heterogeneous agents within the financial system.\\

In general, this research demonstrates the power of the MFG framework in analyzing heterogeneous agent models and provides a deeper understanding of how agent heterogeneity, financial frictions, and mean-field interaction all contribute to the phenomenon of slow economic recovery and the deviating behavior of various agents in the economy during different economic phases.

\section{A Macro-Finance Game Without Aggregate Shocks}
In this section, we shall build a simple framework where macro-noise does not have an explicit effect on the capital process of both experts and households. However, it should be noted that aggregate shocks have an implicit effect on the capital process through the mean-field interaction with other agents in the economy.

\subsection{A Finite Players Game}

We construct an economy as a competitive game among $N$ agents in the context of macro-finance and approximate the Nash equilibrium of this economy using tools from the mean field game. First, consider two heterogeneous groups in the economy, namely {\it households} and {\it experts}. Let $N_e \in \mathbb N$, and $N_h \in \mathbb N$ define the number of agents in the groups of experts and households, respectively, with $N_e + N_h = N$. Also, denote $k_{e,t}^i$ as the capital of expert $i$, $i = 1, \ldots , N_e$ at time $t$; and $k_{h,t}^i$ as the capital of household $i$, $i = 1, \ldots , N_h$ at time $t$. Next, we denote the group average capitalization of experts and households, respectively, as 
\begin{equation} \label{eq: average exp-house}
\widehat{k}_{e,t}=\frac{k_{e,t}^1+ \cdots +k_{e,t}^{N_e}}{N_e}, \quad \widehat{k}_{h,t}=\frac{k_{h,t}^1+ \cdots +k_{h,t}^{N_h}}{N_h} ; \quad t \ge 0 .  
\end{equation}
Also, the overall average capitalization is defined by 
$$\overline{k}_t=\frac{1}{N_h+N_e} \,  \bigg(k_{e,t}^1+ \cdots +k_{e,t}^{N_e}+k_{h,t}^1+ \cdots +k_{h,t}^{N_h}\bigg); \quad t \ge 0 . 
$$ 
Thus, with the overall average capitalization $\overline{k}_\cdot$, the relative average capitalization can be written as:
\begin{equation} \label{eq: RACap}
\begin{split}
    \widehat{k}_t^e & =(1-\lambda_e)\widehat{k}_{e,t}+\lambda_e\overline{k}_t, \\
    \widehat{k}_t^h & =(1-\lambda_h)\widehat{k}_{h,t}+\lambda_h\overline{k}_t,
\end{split}
\end{equation}
for $t \ge 0 $, where $\lambda_e,\lambda_h\in [0, 1] $ are hyper-parameters to be chosen and each of them is the relative preference for tracking one's group average as opposed to the global average. We shall use these $\widehat{k}^e_\cdot$ and $\widehat{k}^h_\cdot$ in the mean-field game analysis. \\

We are interested in exploring the behavior of the general economy as a whole, once the number of agents in such an economy becomes significantly large (i.e. $N_e,N_h\to \infty$). In other words, the finite-player games turn into mean-field games, and the equilibrium will be well characterized by the stochastic Pontryagin maximum principle. We will discuss this in detail in the following sections.\\

Now, suppose that on a filtered probability space $ (\Omega, \mathcal{F}, \mathbb{P}, \{\mathcal{F}_t\}) $ each capital $k^i_{e, t}$ of the expert evolves as an Ornstein-Uhlenbeck (OU) process with the reverting mean as follows: for $t \ge 0 $, 
\begin{align}
    dk_{e,t}^i=[\widehat{k}_{t}^e-|\Phi(\iota_t^i)-\delta_e|\cdot k_{e,t}^i]dt+\sigma_{e,k}^idB^e_t,  \quad \quad i = 1,2, \ldots ,N_e,
    \label{3.1.1}
\end{align}
where $ B^e_t $, $t \ge 0 $ denotes a $ 1 $-dimensional Wiener process, and $\sigma_{e,k}^i$ is a constant representing the volatility of the endogenous risk created by the experts.\\

The function $\Phi(\cdot)$ represents the cost of capital accumulation from the investment strategies of the experts, while $\iota_t^i$ denotes the internal investment per capital rate at time $t$ of agent $i$ (that is, $\iota_t^ik_{e,t}^i$ becomes the investment rate). We assume that $\Phi(0)=0$, and $\Phi(\cdot)$ is an increasing and concave function. Without new investment, capital depreciates at rate $\delta_e$. Experts take more capital when the relative average capital $\widehat{k_t^e}$ is sufficiently large compared to the current amount of capital and vice versa. In other words, they are confident in accumulating more capital when their counterparts are doing the same thing. \\

To be more precise, the drift term $d:=[\widehat{k}_{t}^e-|\Phi(\iota_t^i)-\delta_e|\cdot k_{e,t}^i]$ in \eqref{3.1.1} is assumed to be global Lipschitz continuous functions (i.e. there exists a constant $C>0$ such that $|d(t,x)-d(t,x')|\leq C\cdot |x-x'|$, where $x,x'\in\mathbb R^d$ and $t\in[0,T]$) and linear growth (i.e. there exists a constant $C>0$ such that $|d(t,x)|\leq C(1+|x|)$) in $x:=k_{e,t}^i$ so that the strong solution to \eqref{3.1.1} uniquely exists (see, e.g., Karatzas and Shreve (2014)). In this paper, we assume such conditions for all SDEs. \\

Similarly, the dynamics of household capital $k_{h,t}^i$ is described by the following SDE: 
\begin{align}
    dk_{h,t}^i=(\widehat{k}_{t}^h-(1-\delta_h)k_{h,t}^i)dt+\sigma_{h,k}^idB^h_t,  \quad \quad i = 1,2,...,N_h,
\end{align}
for $t \ge 0 $, where $ B^h_\cdot $ denotes a $ 1 $-dimensional Wiener process on the filtered probability space $ (\Omega, \mathcal{F}, \mathbb{P}, \{\mathcal{F}_t\}) $; and $\sigma_{e,k}^i$ is a constant representing the volatility of the endogenous risk created by the experts. In addition, $\delta_h$ is the capital adjustment speed of households such that $\delta_h>\delta_e$. Although households do not have internal investment technology like experts to slow the rate down, they can freely trade capital on a liquid market to earn a risk-free rate. That is, they will exchange more capital when the relative average capital $\widehat{k_t^h}$ is sufficiently large compared to the current amount of capital held and vice versa.  \\

The price $q_t$ of capital  is driven by the following stochastic differential equation:
\begin{align}
dq_t=\mu_t^qq_t dt+\sigma_t^qq_tdW_t, \quad t \ge 0 , 
\end{align}
where $ W_t $ denotes a $ 1 $-dimensional Wiener process on the  filtered probability space $ (\Omega, \mathcal{F}, \mathbb{P}, \{\mathcal{F}_t\}) $, with an initial value $q_0$, and $\mu_t^q$ and $\sigma_t^q$ are constants. Moreover, $B^e_t$, $B^h_t$, and $W_t$ are the idiosyncratic income shocks to experts, households, and the aggregate market shocks, respectively. We assume that they have the following correlations:
\begin{align}
    \langle B^e,W\rangle_t=p_1\cdot t , \quad  \langle B^h,W\rangle_t=p_2\cdot t, \quad \langle B^h,B^e\rangle_t=p_3\cdot t, \quad t \ge 0  
\end{align}
in terms of the cross variation with $p_1 \in[-1,1]$, $p_2\in [-1,1]$, and $p_3 \in[-1,1]$ representing the instantaneous correlation coefficients.\\

Using Ito's formula, we write down the dynamics of the experts' capital gains as follows:
\begin{equation}
\begin{split}
    d(k_{e,t}^iq_t)=\bigg(q_t(\hat{k_{t}^e}-|\Phi(\iota_t^i)-\delta_e|\cdot k_{e,t}^i)+k_{e,t}^i\mu_t^qq_t+p_1\sigma_t^q\sigma^i_{e,k} q_t\bigg)dt\\
    +q_tk_{e,t}^i\sigma_t^qdW_t+q_t\sigma^i_{e,t}dB^e_t,      
\end{split}
\end{equation}
for $t\geq 0$. Similarly, we have the households' capital gains dynamics:
\begin{equation}
\begin{split}
     d(k_{h,t}^iq_t)=\bigg(q_t(\widehat{k}_{t}^h-(1-\delta_h)k_{h,t}^i)+k_{h,t}^i\mu_t^qq_t+p_2\sigma_t^q\sigma^i_{h,k} q_t\bigg)dt\\
+q_tk_{h,t}^i\sigma_t^qdW_t+q_t\sigma^i_{h,t}dB^h_t,
\end{split}
\end{equation}
where $t\geq 0$. For our convenience, we set for $t \ge 0$:
\begin{align}
    u_1(t)& := q_t(\widehat{k}_{t}^e-|\Phi(\iota_t^i)-\delta_e|\cdot k_{e,t}^i)+k_{e,t}^i\mu_t^qq_t+p_1\sigma_t^q\sigma^i_{e,k} q_t,\\
    u_2(t)&:=q_t(\widehat{k}_{t}^h-(1-\delta_h)k_{h,t}^i)+k_{h,t}^i\mu_t^qq_t+p_2\sigma_t^q\sigma^i_{h,k} q_t.
\end{align}

Now, suppose that experts hold a fraction $\phi (\geq \tilde{\phi)}$ of capital risk and unload the rest to less productive households through equity issuance. Then the law of motion of experts' wealth $w_{e, \cdot}^i$ is given by 
\begin{align} \label{13} 
    dw_{e,t}^i=rw_{e,t}^idt+\bigg(u_1(t)-k_{e,t}^iq_tr-c_{e,t}^i\bigg)dt+\phi q_tk_{e,t}^i\sigma_t^qdW_t+q_t\sigma^i_{e,t}dB^e_t,
\end{align}
for $t\geq 0$, $i = 1, \ldots , N_e$, where $r > 0 $ is the risk-free interest rate, and $c_{e,t}^i\geq 0$ is the consumption rate of the expert. Note that the fraction $\phi$ represents the financial friction in the form of an equity constraint here and hinders the movement of capital between agents. The idea is that experts can issue some outside equity and transfer the risk to less productive households. The fraction must always be larger than the lower bound $\tilde{\phi}$, the threshold to distinguish households and experts. In addition, experts can invest in capital only when their net worth is positive, which means that they are limited by both an equity constraint and a solvency constraint. \\

Each expert $i$ will maximize their expected utility, by controlling the consumption rate $c_{e,t}^i$, the capital holding $k_{e,t}^i$, and the investment rate $\iota_t^i$:
\begin{align}
    J_e^i=\max_{c_{e,t}^i\geq 0,k_{e,t}^i\geq 0,\iota_t^i}\mathbb E\bigg[\int_0^T e^{-\rho t}U(c_{e,t}^i)dt +w_{e,T}^i\bigg],\label{14}
\end{align}
where $\rho$ is the discount rate ($\rho>r$), $U(.)$ is the utility function, and $w_{e,T}^i$ is the terminal wealth of the expert at the end period referred to the wealth process in Equation \eqref{13}. In this model, we choose $U(.)$ to be the utility function of CRRA type with the risk aversion parameter $\gamma ( > 0) $, that is, 
\begin{align}
    U(c)=\frac{c^{1-\gamma}-1}{1-\gamma},\quad  \quad c > 0 \, . 
\end{align}

Next, similarly, we have the household wealth dynamics:
\begin{align}
    dw_{h,t}^i=rw_{h,t}^idt+\bigg(u_2(t)-k_{h,t}^iq_tr-c_{h,t}^i\bigg)dt +(1-\phi)q_tk_{h,t}^i\sigma_t^qdW_t +q_t\sigma^i_{h,t}dB^h_t,
\end{align}
for $t\geq 0$. Each household $i$ will maximize their expected utility, by controlling the consumption rate $c_{h,t}^i$, the capital holding $k_{h,t}^i$ :
\begin{align}
    J_h^i=\max_{c_{h,t}^i,k_{h,t}^i>0}\mathbb E\bigg[\int_0^T e^{-rt}U(c_{h,t}^i)dt+w_{h,T}^i\bigg].\label{17}
\end{align}
where $r$ is the risk-free discount rate, and we allow $c_{h,t}^i$, the consumption, to be either positive or negative. \\

Next, for our convenience, we further set: \begin{equation}
\begin{split}
b_e(w_{e,t}^i,q_t,\widehat{k_t^e},k_{e,t}^i,c_{e,t}^i)& =rw_{e,t}^i+q_t(\widehat{k_{t}^e}-|\Phi(\iota_t^i)-\delta_e|\cdot k_{e,t}^i) \\
& \qquad {}+ k_{e,t}^i\mu_t^qq_t+p_1\sigma_t^q\sigma^i_{e,k} q_t-k_{e,t}^iq_tr-c_{e,t}^i,
\end{split}
\end{equation}

\begin{equation}
    \begin{split}
b_h(w_{h,t}^i,q_t,\widehat{k_t^h},k_{h,t}^i,c_{h,t}^i)& =rw_{h,t}^i+q_t(\widehat{k_{t}^h}-(1-\delta_h)k_{h,t}^i)\\
& \qquad {}+k_{h,t}^i\mu_t^qq_t+p_2\sigma_t^q\sigma^i_{h,k} q_t-k_{h,t}^iq_tr-c_{h,t}^i, 
    \end{split}
\end{equation}

\begin{equation}
    \begin{split}
        v_e(\phi,q_t,k_{e,t}^i)& =\phi q_tk_{e,t}^i\sigma_t^q, \qquad 
v_h(\phi,q_t,k_{h,t}^i)=(1-\phi)q_tk_{h,t}^i\sigma_t^q.
    \end{split}
\end{equation}
Then the law of motion of expert's wealth $w_{e,\cdot}^i$ and of household wealth $w_{h,\cdot}^i$ can be rewritten as 
\begin{align}
dw_{e,t}^i=b_e(w_{e,t}^i,q_t,\widehat{k_t^e},k_{e,t}^i,c_{e,t}^i)dt+v_e(\phi,q_t,k_{e,t}^i)dW_t+q_t\sigma^i_{e,t}dB^e_t 
\end{align}
for $i = 1, \ldots , N_e$, and 
\begin{align}  dw_{h,t}^i=b_h(w_{h,t}^i,q_t,\widehat{k_t^h},k_{h,t}^i,c_{h,t}^i) dt+v_h(\phi,q_t,k_{h,t}^i)dW_t+q_t\sigma^i_{h,t}dB^h_t
\end{align}
for $i = 1, \ldots , N_h$, $t \ge 0 $, respectively. 

\subsection{Macro Mean-field Games}

Observe that the state dynamics follow a Markov diffusion (i.e. both drift and volatility terms are deterministic functions, which only depend on the present value $w^i_t$). Hence, we attempt to solve an approximate Markovian Nash equilibrium in the mean-field limit sense described in Definition 2.10 of Carmona \& Delaure (2018) as $N:=N_e+N_h\to \infty$. The existence and uniqueness of the MFG solution are the existence and uniqueness of the FBSDE solution by the Pontryagin Stochastic Maximum Principle.\\

First, for experts, we fix $m_{e,t}=\mathbb E[k_{e,t}|W_{s,s\leq t}]$ which is a candidate for the limit of $\widehat{k}_{e,t}$ in \eqref{eq: RACap} when $N_e\to \infty$: $$m_{e,t}=\lim_{N_e\to\infty}\widehat{k}_{e,t}.$$
Similarly, for households, we have:
$$m_{h,t}=\lim_{N_h\to\infty}\widehat{k}_{h,t}.$$
And, also the global average is:
$$M_t=\lim_{N_e,N_h\to\infty}\frac{N_e}{N_e+N_h}\widehat{k}_{e,t}+\frac{N_h}{N_e+N_h}\widehat{k}_{h,t}.\label{22.1}$$
Now, the relative average becomes:
$$M_{e,t}=(1-\lambda_e)m_{e,t}+\lambda_eM_t,\label{22.2}$$
$$M_{h,t}=(1-\lambda_h)m_{h,t}+\lambda_hM_t.\label{22.3}$$

Next, we solve the fixed-point problem by constructing the FBSDEs and using the Pontryagin stochastic maximum principle.
\begin{proposition}
     The optimal consumption of experts and households is given by:
    \begin{align}
    \widehat{c}_{e,t}=e^{-\frac{\rho t}{\gamma}}(y_e)^{-\frac{1}{\gamma}},
\end{align}
\begin{align}
    \widehat{c}_{h,t}=e^{-\frac{r t}{\gamma}}(y_h)^{-\frac{1}{\gamma}},
\end{align}
where $y_e=y_h=e^{-r(T-t)}$ for $0\leq t\leq T$.
\end{proposition}
\begin{proof}
See Appendix \ref{a21}
\end{proof}

Next, the equilibrium of the economy are solved explicitly through the following theorem.
\begin{theorem}
     Assuming that $\lambda_e,\lambda_h \in [0,1]$, and $\alpha_e,\alpha_h,\delta_e,\delta_h \in \mathbb R^+$. Also, set 
    \begin{align*}
    A_t=\begin{pmatrix}
    1-\lambda_e+\lambda_e\alpha_e-|\Phi(\iota_t)-\delta_e| & \lambda_e\alpha_h\\
    \lambda_h\alpha_e & -\lambda_h+\lambda_h\alpha_h+\delta_h
    \end{pmatrix},
\end{align*}
 $$L=\begin{pmatrix}\sigma_{e,k}&0 \\ 0&\sigma_{h,k} \end{pmatrix}.$$
    
Then, an approximate Nash equilibrium is given by:
\begin{align}
    m_t=v_t\bigg[m_0+\int_{0}^tv^{-1}_sL dB^l_s\bigg]; \quad 0\leq t\leq T,
\end{align}
where $v_t:=e^{\int_0^\infty A_sds}$ is the fundamental solution of the following homogeneous equation:
\begin{align*}
    dv_t=A_tv_tdt; \quad v_0=I
\end{align*}
where $I$ is the $2\times 2$ identity matrix, and 
\begin{align}
    B^l_t=\begin{pmatrix} B^e_t \\ B^h_t \end{pmatrix}.
\end{align}
\end{theorem}

\begin{proof}
     See Appendix \ref{a22}
\end{proof}

\subsection{Numerical Simulation}

After having the flow of measure, we can plug it into the forward equations of the capital and wealth processes to solve for the other MFG equilibrium. We perform a numerical simulation to illustrate our results.\\

Suppose that the investment cost function $\Phi(\iota_t)$ takes the following logarithmic form:
\begin{align}
    \Phi(x)=\frac{\log (\kappa x+1)}{\kappa}, \quad x > 0 . 
\end{align}

where $\kappa$ is the adjustment cost parameter that controls the elasticity of the investment technology. \\

When the function $\Phi(\iota_t)$ is not deterministic, we can solve the optimal investment rate based on the results of q-theory:

\begin{align}
    \iota_t^*=\frac{q_t-1}{\kappa}.
\end{align}

Hence, the capital price through the investment rate determines the economy's growth rate endogenously. An increase in pricing stimulates investment and speeds up production growth. Now, the investment cost function becomes:
\begin{align}
    \Phi(\iota_t^*)=\frac{\log (q_t)}{\kappa}.
\end{align}

\begin{figure}[H]
    \centering
    \includegraphics[scale=0.55]{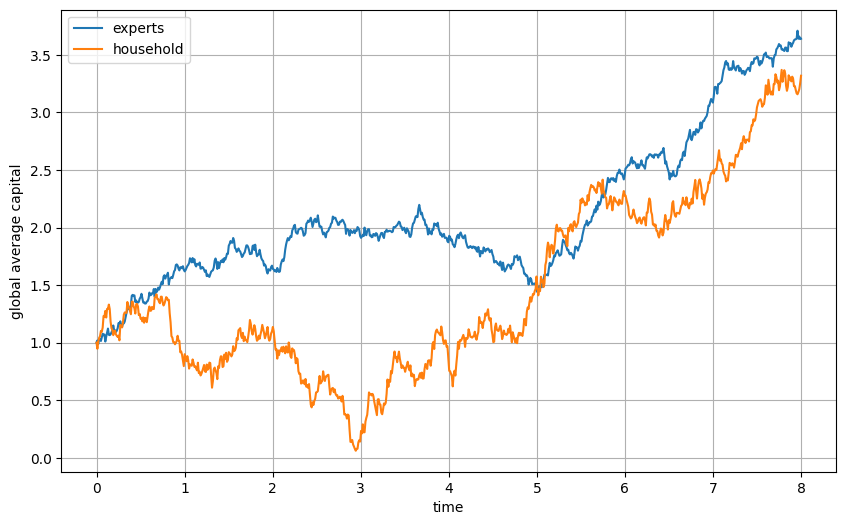}
    \caption{Global Average Capital Process}
\end{figure}

We use the Euler-Maruyama scheme to model the capital price, which follows a geometric Brownian motion, over a period of $800$ weeks: 

\begin{figure}[H]
    \centering
    \includegraphics[scale=0.65]{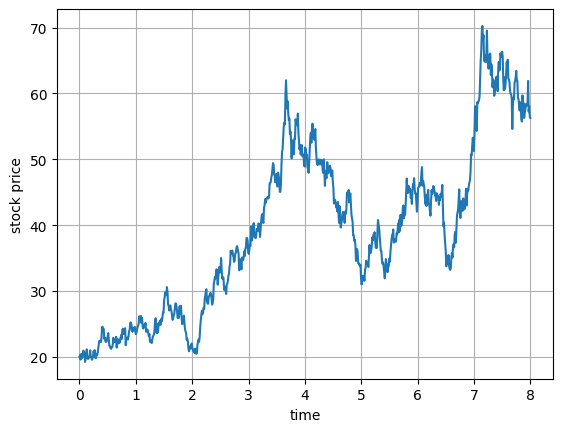}
    \caption{Price of Capital}
\end{figure}

Our model exhibits different phases of economic activity. Two prominent boom periods emerge, one spanning from the initial period to week 350, and another from week 650 to week 720. In contrast, a prolonged crisis phase is observed between weeks 400 and 620, with the potential for another crisis following week 720. We focus on the contrasting behavior of households and experts during these economic regimes. \\

Households exhibit a clear pattern of reducing their capital holdings during prosperous periods. This behavior is based on the recognition of the superior capital management skills of experts during booms. By delegating capital to experts, households optimize their own returns. The situation reverses during economic downturns, when experts are forced to engage in fire sales of capital assets. Households, perceiving an opportunity, strategically increase their capital holdings, essentially acting as a lifeline for struggling firms during these periods of stress. Once experts recover their footing, households gradually relinquish control by decreasing their capital holdings once again.\\

Experts, on the other hand, prioritize capital accumulation during high-growth phases. This strategy serves a dual purpose: wealth creation and the establishment of a buffer against potential future economic downturns. However, their overconfidence during booms leads to the accumulation of significant capital holdings. This strategy becomes a liability at the onset of crises. The interaction of the mean field with other agents in the economy forces experts to engage in asset sales, often at the fire-sale prices. The combination of overconfidence and excessive capital accumulation during booms results in significant losses for experts when economic conditions deteriorate, and other agents, seeking to protect their positions, also begin selling capital.\\

\begin{figure}[H]
    \centering
    \includegraphics[scale=0.55]{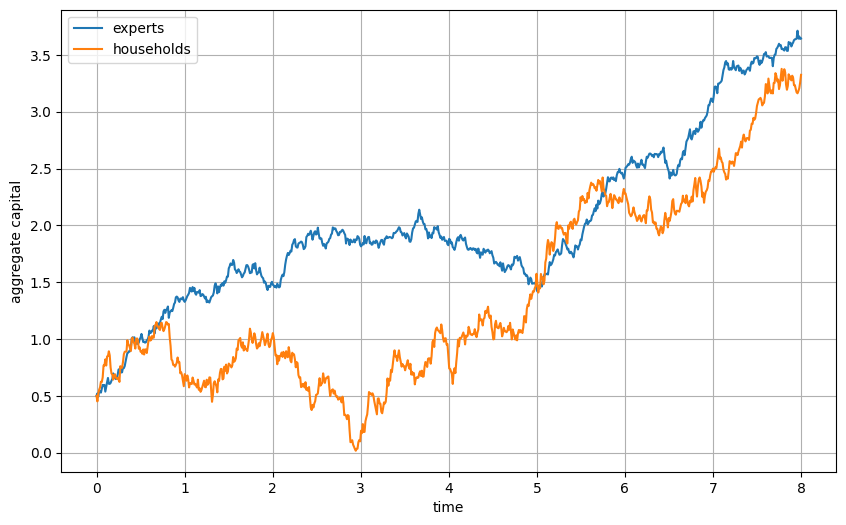}
    \caption{The Aggregate Capital of Experts and Households}
\end{figure}
 
This analysis highlights the dynamic interplay between heterogeneous agents within the model. Households and experts exhibit strategic behavior that is not only self-serving but also responsive to the actions of others through mean-field interaction. This dynamic co-evolution of behavior contributes to the observed boom-bust cycles and the slow recovery process following economic crises.\\

\begin{figure}[H]
    \centering
    \includegraphics[scale=0.55]{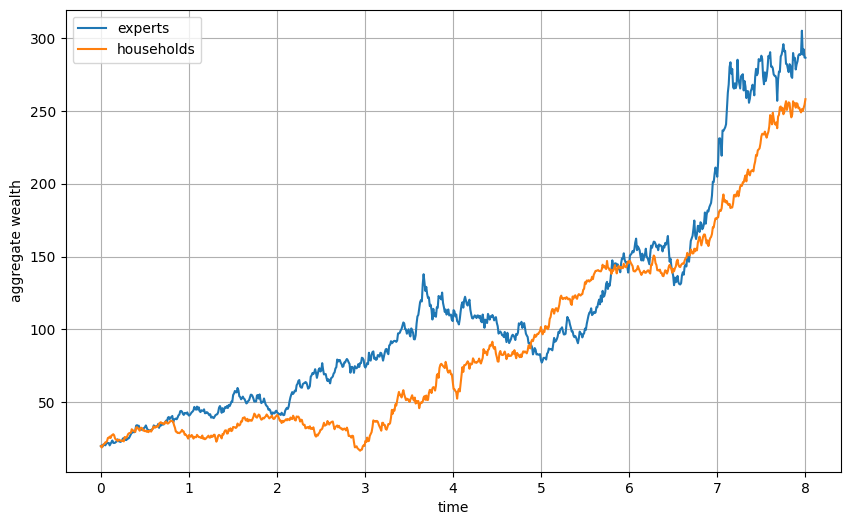}
    \caption{The Wealth Distribution}
\end{figure}   

\section{A Macro-Finance Game With Aggregate Shocks}

In this section, we will consider a general model in which we allow common noise from the market to explicitly affect both households' and experts' capital processes in addition to their income shocks.
\subsection{A Finite Players Game}

Suppose that the capital of the expert evolves as an OU process with the reverting mean as follows:
\begin{align}
    dk_{e,t}^i=[\widehat{k_{t}^e}-|\Phi(\iota_t^i)-\delta_e|\cdot k_{e,t}^i]dt+\sigma_{e,k}^idB^e_t+\sigma dW_t \quad \quad i\in\{1,2,...,N_e\},
    \label{3.2.1}
\end{align}
where $B^e_t$ and $ W_t $ denote $ 1 $-dimensional Wiener processes on a filtered probability space $ (\Omega, \mathcal{F}, \mathbb{P}, \{\mathcal{F}_t\}) $. 
In addition, $\sigma_{e,k}^i \in \mathbb R$ is the volatility of the endogenous shocks of the expert, and $\sigma \in \mathbb R$ is the volatility of the market.\\

Similarly, we have the dynamics of household capital is described by the following SDE:
\begin{align}
    dk_{h,t}^i=(\widehat{k_{t}^h}-(1-\delta_h)k_{h,t}^i)dt+\sigma_{h,k}^idB^h_t+\sigma dW_t \quad \quad i\in\{1,2,...,N_h\},
\end{align}
where $B^h_t$ represents the household's individual income shocks. In addition, $\sigma_{h,k}^i\in \mathbb R$ is the volatility of endogenous shocks in the household.\\

Using Ito's formula, we write down the dynamics of the experts' capital gains as follows:
\begin{align*}
    d(k_{e,t}^iq_t)=\bigg(q_t(\widehat{k_{t}^e}-|\Phi(\iota_t^i)-\delta_e|\cdot k_{e,t}^i)+k_{e,t}^i\mu_t^qq_t+\sigma_t^q\sigma q_t+p_1\sigma_t^q\sigma_{e,k}^iq_t\bigg)dt\\
+q_t\sigma_{e,k}^idB_t^e+q_t(k_{e,t}^i\sigma_t^q+\sigma)dW_t,
\end{align*}
for $t\geq 0$.\\

Similarly, we have the households' capital gains dynamics:
\begin{align*}
     d(k_{h,t}^iq_t)=\bigg(q_t(\widehat{k_{t}^h}-(1-\delta_h)k_{h,t}^i)+k_{h,t}^i\mu_t^qq_t+\sigma_t^q\sigma q_t+p_2\sigma_t^q\sigma_{h,k}^iq_t\bigg)dt\\
+q_t\sigma_{h,k}^idB_t^h+q_t(k_{h,t}^i\sigma_t^q+\sigma)dW_t,
\end{align*}
for $t\geq 0$.\\

Next, for our convenience, set:
\begin{align}    b_e^k(q_t,\widehat{k_t^e},k_{e,t}^i,c_{e,t}^i)=q_t(\hat{k_{t}^e}-|\Phi(\iota_t^i)-\delta_e|\cdot k_{e,t}^i)+k_{e,t}^i\mu_t^qq_t+\sigma_t^q\sigma q_t+p_1\sigma_t^q\sigma_{e,k}^iq_t,
\end{align}
\begin{align}  b_h^k(q_t,\widehat{k_t^e},k_{e,t}^i,c_{e,t}^i)=q_t(\hat{k_{t}^h}-(1-\delta_h)k_{h,t}^i)+k_{h,t}^i\mu_t^qq_t+\sigma_t^q\sigma q_t+p_2\sigma_t^q\sigma_{h,k}^iq_t.
\end{align}

Now, suppose that experts hold a fraction of $\phi\geq \tilde{\phi}$ of capital risk and unload the rest to less productive households through equity issuance, the law of motion of experts' wealth is:
\begin{equation}
\begin{split}
dw_{e,t}^i=rw_{e,t}^idt+\bigg(b_e^k(q_t,\widehat{k_t^e},k_{e,t}^i,c_{e,t}^i)-k_{e,t}^iq_tr-c_{e,t}^i\bigg)dt\\
+q_t\sigma_{e,k}^idB_t^e+\phi q_t(k_{e,t}^i\sigma_t^q+\sigma)dW_t,
\end{split}
\end{equation}
for $t\geq 0$.\\

Next, similarly, we have the household wealth dynamics:
\begin{equation}
\begin{split}
dw_{h,t}^i=rw_{h,t}^idt+\bigg(b_h^k(q_t,\widehat{k_t^h},k_{h,t}^i,c_{h,t}^i)-k_{h,t}^iq_tr-c_{h,t}^i\bigg)dt\\
+q_t\sigma_{h,k}^idB_t^h+(1-\phi)q_t(k_{h,t}^i\sigma_t^q+\sigma)dW_t.
\end{split}
\end{equation}
for $t\geq 0$.\\

Each of the experts maximizes their expected utility in the form of Equation \eqref{14}, while each of the households maximizes their expected utility in the form of Equation \eqref{17}.\\

Next, for our convenience, set: 
\begin{align}
b_e^w(w_{e,t}^i,q_t,\widehat{k_t^e},k_{e,t}^i,c_{e,t}^i)=rw_{e,t}^i+b_e^k(q_t,\widehat{k_t^e},k_{e,t}^i,c_{e,t}^i)-k_{e,t}^iq_tr-c_{e,t}^i,
\end{align}

\begin{align}
b_h^w(w_{h,t}^i,q_t,\widehat{k_t^h},k_{h,t}^i,c_{h,t}^i)=rw_{h,t}^i+b_h^k(q_t,\widehat{k_t^h},k_{h,t}^i,c_{h,t}^i)-k_{h,t}^iq_tr-c_{h,t}^i,
\end{align}

\begin{equation}
    \begin{split}
       v_e(\phi,q_t,k_{e,t}^i)&=\phi q_t(k_{e,t}^i\sigma_t^q+\sigma), \qquad 
v_h(\phi,q_t,k_{h,t}^i)=(1-\phi)q_t(k_{h,t}^i\sigma_t^q+\sigma).
    \end{split}
\end{equation}


The law of motion of expert's wealth becomes:
\begin{align}
dw_{e,t}^i=b_e^w(w_{e,t}^i,q_t,\widehat{k_t^e},k_{e,t}^i,c_{e,t}^i)dt+q_t\sigma_{e,k}^idB_t^e+v_e(\phi,q_t,k_{e,t}^i)dW_t,
\end{align}
for $t\geq 0$.\\

And, the law of motion of household's wealth becomes:
\begin{align}
dw_{h,t}^i=b_h^w(w_{h,t}^i,q_t,\widehat{k_t^h},k_{h,t}^i,c_{h,t}^i) dt+q_t\sigma_{h,k}^idB_t^h+v_h(\phi,q_t,k_{h,t}^i)dW_t.
\end{align}
for $t\geq 0$.\\

\subsection{Macro Mean-field Games}

 Again, we attempt to solve an approximate Markovian Nash equilibrium in the mean-field limit described in Definition 2.10 of Carmona \& Delaure (2018) as $N:=N_e+N_h\to \infty$. \\

\begin{proposition}
    The optimal consumption of experts and households are given by:
    \begin{align}
    \widehat{c_{e,t}}=e^{-\frac{\rho t}{\gamma}}(y_e)^{-\frac{1}{\gamma}},
\end{align}
\begin{align}
    \widehat{c_{h,t}}=e^{-\frac{r t}{\gamma}}(y_h)^{-\frac{1}{\gamma}},
\end{align}
where $y_e=y_h=e^{-r(T-t)}$ and $0\leq t\leq T$.   
\end{proposition}
\begin{proof}
See Appendix \ref{a31}
\end{proof}

Next, the equilibria of the economy are solved explicitly through the following theorem. 
\begin{theorem}
     Assuming that $\lambda_e,\lambda_h \in [0,1]$, and $\alpha_e,\alpha_h,\delta_e,\delta_h \in \mathbb R^+$. Also, set 
    \begin{align*}
    A_t=\begin{pmatrix}
    1-\lambda_e+\lambda_e\alpha_e-|\Phi(\iota_t)-\delta_e| & \lambda_e\alpha_h\\
    \lambda_h\alpha_e & -\lambda_h+\lambda_h\alpha_h+\delta_h
    \end{pmatrix},
\end{align*}
 $$S=\begin{pmatrix}\sigma&0 \\ 0&\sigma \end{pmatrix},$$ and
 $$L=\begin{pmatrix}\sigma_{e,k}&0 \\ 0&\sigma_{h,k} \end{pmatrix}.$$
    
Then, an approximate Nash equilibrium is given by:
\begin{align}
    m_t=v_t\bigg[m_0+\int_{0}^tv^{-1}_sL dB^l_s+\int_{0}^tv^{-1}_sS dW^l_s\bigg]; \quad 0\leq t\leq T
\end{align}
where $v_t:=e^{\int_0^\infty A_sds}$ is the fundamental solution of the following homogeneous equation:
\begin{align*}
    dv_t=A_tv_tdt; \quad v_0=I
\end{align*}
where $I$ is the $2\times 2$ identity matrix, and
\begin{align}
    B^l_t&=\begin{pmatrix} B^e_t \\ B^h_t \end{pmatrix},\\
    W^l_t&=\begin{pmatrix} W_t \\ W_t \end{pmatrix}.
\end{align}

\end{theorem}
\begin{proof}
 See Appendix \ref{a32}.
\end{proof}

\subsection{Numerical Simulation}
After having the flow of measure, we can plug it into the forward equations of the capital and wealth processes to solve for the other MFG equilibrium. We perform a numerical simulation to illustrate our results.\\

\begin{figure}[H]
    \centering
    \includegraphics[scale=0.55]{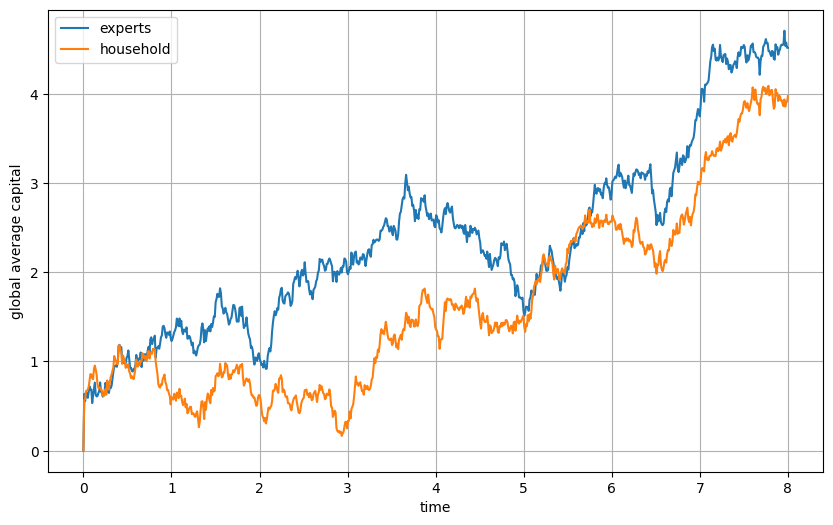}
    \caption{Global Average Capital Process}
\end{figure}

The introduction of aggregate macroeconomic shocks into our model alters the behavior of both households and experts. Households exhibit increased vigilance in response to these aggregate shocks. Their capital accumulation during crises becomes more cautious, characterized by a slower rate of increase. Their small proportion of market risk sharing $(1-\tilde{\phi})$ acts as a buffer, shielding households from the most severe market fluctuations. This prevents them from resorting to fire sales of their capital holdings, allowing them to continue fulfilling their role as a safety net for firms during economic downturns.\\

In contrast, expert behavior undergoes a significant modification. Their focus on capital accumulation during high-growth phases diminishes in favor of prioritizing stable productivity. This shift reflects their increased awareness of the potential risks associated with large capital positions in the face of macroeconomic shocks. However, the overconfidence exhibited during good times, combined with the mean-field interaction, still forces experts to engage in asset sales at the onset of crises. This dynamic highlights the persistence of vulnerabilities even with a more cautious approach to capital accumulation.\\

\begin{figure}[H]
    \centering
    \includegraphics[scale=0.55]{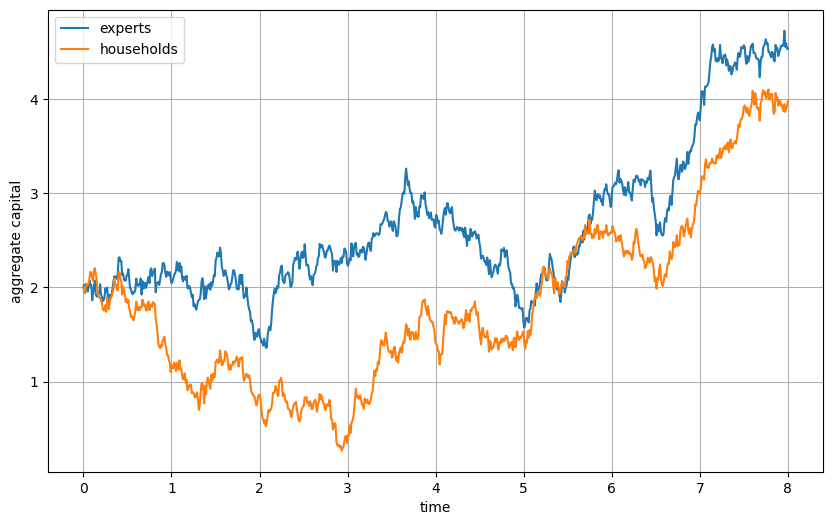}
    \caption{The Aggregate Capital of Experts and Households}
\end{figure}

The accumulation of wealth in households exhibits minimal growth during crisis periods. This phenomenon can be attributed to their strategic acquisition of capital from experts. By absorbing capital assets divested by experts at fire-sale prices, households essentially function as a financial backstop for struggling firms during economic downturns. This intervention, while stabilizing the system, comes at the expense of immediate wealth creation for households. However, there is the potential for significant returns once the crisis subsides. Households can then strategically sell back the acquired capital to experts, potentially capturing substantial profits as the market recovers. This highlights the risk-return trade-off inherent in the behavior of households within the model. Their actions contribute to systemic stability during crises, but require delayed wealth accumulation.\\

\begin{figure}[H]
    \centering
    \includegraphics[scale=0.55]{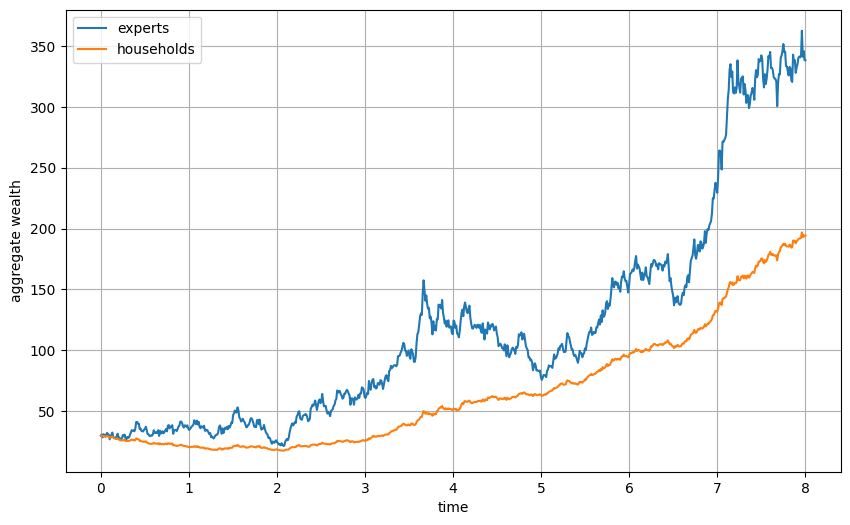}
    \caption{The Wealth Distribution}
\end{figure}

These observations further emphasize the complex interaction between agent heterogeneity, financial frictions, and macroeconomic shocks in shaping economic dynamics. Although the presence of financial frictions can offer some protection for experts to a $\tilde{\phi}$ level, it does not eliminate the underlying vulnerabilities associated with expert overconfidence and mean field interactions. This interplay of factors contributes to the persistence of boom-bust cycles and the challenges associated with achieving rapid economic recovery.

\section{A General Model}
\subsection{A Finite Player Game}
In this section, we will build a general heterogeneous macro-finance model with multiple groups of interacting agents. Suppose that there are $n$ groups of interest on the market and $N_j$ define the number of agents that share similar characteristics in the group $j=1,2,...,n$ such that $N:=\sum_{j=1}^n N_j$. Also, denote $k^{i}_{j,t}$ as the capital of agent $i$, $i=1,...,N_j$ in group $j$ at time $t$. Next, we define the group average capitalization, the overall average capitalization, and the relative average capitalization respectively as follows: 
\begin{align}
    \widehat{k}_{j,t}=\sum_{i=1}^{N_j}\frac{k^i_{j,t}}{N_j}; \quad \overline{k}_t=\sum_{j=1}^n\sum_{i=1}^{N_j}\frac{k^i_{j,t}}{N}; \quad \widehat{k}^{j}_t=(1-\lambda_j)\widehat{k}_{j,t}+\lambda_j\overline{k}_t
\end{align}
where $\lambda_j \in [0,1]$ is the relative preference for tracking one's group average as opposed to the global average.\\

We are interested in exploring the behavior of the general economy as a whole, once the number of agents in such an economy becomes significantly large (i.e. $N_j\to \infty$). In other words, the finite-player games turn into mean-field games, and the equilibrium will be well characterized by the stochastic Pontryagin maximum principle. We will discuss this in detail in the following sections.\\

Now, suppose that the capital process $k_{j,t}^i$ of agent $i$, $i=1,2,...,N_j$, in group $j$, $j=1,2,...,n$, evolves as an OU process with the reverting mean as follows, for $t\geq 0$,
\begin{align}
    dk^i_{j,t}=[\widehat{k}^j_t-|\Phi(\iota^i_{j,t})-\delta_j|\cdot k^i_{j,t}]dt+\sigma^i_{j,k}dB^{j}_t+\sigma dW_t, \label{4.1}
\end{align}
where $ B^j_t $ and $W_t$ denote a standard $ 1 $-dimensional Wiener process on a filtered probability space $ (\Omega, \mathcal{F}, \mathbb{P}, \{\mathcal{F}_t\}) $, and represent the idiosyncratic income shocks of group $j$, and the market common noises respectively. We assume that they have the following correlation:
\begin{align}
    \langle B^{j},W\rangle_t=p'_j\cdot t , \quad  \langle B^j,B^o\rangle_t=p_{j,o}\cdot t, 
\end{align}
where $p'_j, \text{ and } p_{j,o}=p_{o,j} \in [-1,1]$.\\

The price of capital $q_t$ is driven by the following stochastic differential equation:
\begin{align}
    dq_t=\mu_t^qq_tdt+\sigma_t^qq_tdW_t, \quad t\geq 0,
\end{align}
with an initial value $q_0$, and $\mu_t^q$ and $\sigma_t^q \in \mathbb R$ are constant.\\

Using Ito's formula, we write down the dynamics of the agents' capital gains as follows:
\begin{align*}
    d(k_{j,t}^iq_t)=\bigg(q_t(\widehat{k_{t}^j}-|\Phi(\iota^i_{j,t})-\delta_j|\cdot k_{j,t}^i)+k_{j,t}^i\mu_t^qq_t+\sigma_t^q\sigma q_t+p'_j\sigma_t^q\sigma_{j,k}^iq_t\bigg)dt\\
+q_t\sigma_{j,k}^idB_t^e+q_t(k_{j,t}^i\sigma_t^q+\sigma)dW_t,
\end{align*}
for $t\geq 0$.\\

Next, for our convenience, set:
\begin{align}    b_j^k(q_t,\widehat{k_t^j},k_{j,t}^i)=q_t(\hat{k_{t}^j}-|\Phi(\iota^i_{j,t})-\delta_j|\cdot k_{j,t}^i)+k_{j,t}^i\mu_t^qq_t+\sigma_t^q\sigma q_t+p'_j\sigma_t^q\sigma_{j,k}^iq_t,
\end{align}
\begin{align}
    u_j^k(\widehat{k}^j_t,\iota^i_{j,t},k^i_{j,t})=\widehat{k}^j_t-|\Phi(\iota^i_{j,t})-\delta_j|\cdot k^i_{j,t}
\end{align}
Now, suppose that agents in the group $j$ hold a fraction of $\phi_j\geq \tilde{\phi_j}$ of capital risk such that $\sum_j\phi_j=1$, the law of motion of agents' wealth now follows:
\begin{equation}
\begin{split}
dw_{j,t}^i=rw_{j,t}^idt+\bigg(b_j^k(q_t,\widehat{k_t^j},k_{j,t}^i,c_{j,t}^i)-k_{j,t}^iq_tr-c_{j,t}^i\bigg)dt\\
+q_t\sigma_{j,k}^idB_t^j+\phi_j q_t(k_{j,t}^i\sigma_t^q+\sigma)dW_t,
\end{split}\label{4.2.1}
\end{equation}
for $t\geq 0$, and $r \in \mathbb R$ is the risk-free rate.\\

Each agent $i$ in the group $j$ will maximize their expected utility, by controlling the consumption rate $c_{j,t}^i$, the capital holding $k_{j,t}^i$, and the investment rate $\iota_{j,t}^i$:
\begin{align}
    J_j^i=\max_{c_{j,t}^i,k_{j,t}^i\geq 0,\iota_{j,t}^i}\mathbb E\bigg[\int_0^T e^{-\rho_j t}U(c_{j,t}^i)dt +w_{j,T}^i\bigg],\label{42}
\end{align}
where $\rho_j$ is the discount rate ($\rho_j\geq r$), $U(.)$ is the utility function, and $w_{j,T}^i$ is the terminal wealth of the agent at the end period referred to the wealth process in Equation \eqref{4.2.1}. In this model, we choose $U(.)$ to be the CRRA utility function with risk-averse parameter $\gamma ( > 0) $ 
, that is, 
\begin{align}
    U(c)=\frac{c^{1-\gamma}-1}{1-\gamma},\quad  \quad c > 0 \, . 
\end{align}

Next, for our convenience, set: 
\begin{align}
b_j^w(w_{j,t}^i,q_t,\widehat{k_t^j},k_{j,t}^i,c_{j,t}^i)=rw_{j,t}^i+b_j^k(q_t,\widehat{k_t^j},k_{j,t}^i,c_{j,t}^i)-k_{j,t}^iq_tr-c_{j,t}^i,
\end{align}

\begin{equation}
    \begin{split}
       v_j(\phi_j,q_t,k_{j,t}^i)&=\phi_j q_t(k_{j,t}^i\sigma_t^q+\sigma).
    \end{split}
\end{equation}

The law of motion of agents' wealth becomes:
\begin{align}
dw_{j,t}^i=b_j^w(w_{j,t}^i,q_t,\widehat{k_t^j},k_{j,t}^i,c_{j,t}^i)dt+q_t\sigma_{j,k}^idB_t^j+v_j(\phi_j,q_t,k_{j,t}^i)dW_t, \label{107}
\end{align}
for $t\geq 0$.\\

\subsection{A Macro Mean Field Game}
We attempt to solve an approximate Markovian Nash equilibrium in the mean-field limit described in Definition 2.10 of Carmona \& Delaure (2018) as $N=\sum_j N_j\to \infty$ and $\frac{N_j}{N}\to \alpha_j$ for all $j$. \\

In doing so, we first fix $m_{j,t}=\mathbb E[k_{j,t}|W_{s,s\leq t}]$ which is a candidate for the limit of $\widehat{k}_{j,t}$ when $N_j\to \infty$: $$m_{j,t}=\lim_{N_j\to\infty}\widehat{k}_{j,t}.$$
And, also the global average is:
$$M_t=\lim_{N_1,N_2,...,N_j\to\infty}\sum_{j=1}^n\frac{N_j}{N}\widehat{k}_{j,t}=\sum_{j=1}^n\alpha_j\cdot m_{j,t}.\label{25.1}$$
Now, the relative average becomes:
$$M_{j,t}=(1-\lambda_j)m_{j,t}+\lambda_jM_t,\label{25.2}$$
Next, we will solve for the $n$-players control problem described by \eqref{42}, and use the solution to solve for the fixed point problem at last.

\begin{proposition}
    The optimal consumption of agents in group $j$ given by:
    \begin{align}
    \widehat{c}_{j,t}=e^{-\frac{\rho_j t}{\gamma}}(y^w_j)^{-\frac{1}{\gamma}},
\end{align}
where $y^w_j=e^{-r(T-t)}$ and $0\leq t\leq T$.   
\end{proposition}
\begin{proof}
See Appendix \ref{a41}.
\end{proof}

Next, the equilibria of the economy are solved explicitly through the following theorem. 
\begin{theorem}
     Assuming that $\lambda_j \in [0,1]$, and $\alpha_j,\delta_j\in \mathbb R^+$ for $j=1,2,...,n$. Also, for a $n\times n$ matrix ${\bm A}_t$, ${\bm L}$, and ${\bm S}$ we have:
     \begin{align}
     \begin{cases}
     A_{ij}=\lambda_i\alpha_j; \quad S_{ij}=0; \quad L_{i,j}=0 \quad \quad \quad \quad \quad \quad \quad \quad \quad \quad \quad \quad \text{ for all } i\neq j, \\
     A_{ii,t}=1-\lambda_i+\lambda_i\alpha_i-|\Phi(\iota_{i,t})-\delta_i|; \quad S_{ii}=\sigma; \quad L{ii}=\sigma_{i,k} \quad \quad \text{ otherwise }.
     \end{cases}
      \end{align}
Then, an approximate Nash equilibrium is given by:
\begin{align}
    m_t=v_t\bigg[m_0+\int_{0}^tv^{-1}_sL dB^l_s+\int_{0}^tv^{-1}_sS dW^l_s\bigg]; \quad 0\leq t\leq T
\end{align}
where a $n\times n$ matrix $v_t:=e^{\int_0^\infty A_sds}$ is the fundamental solution of the following homogeneous equation:
\begin{align*}
    dv_t=A_tv_tdt; \quad v_0=I
\end{align*}
where $I$ is the $n\times n$ identity matrix, and $B^l_t$, and $W^l_t$ are $n$-dimensional Brownian motions such that $B^l_t=(B^1_t,B^2_t,...,B^n_t)$, and $W^l_t=(W_t,W_t,...,W_t)$.
\end{theorem}
\begin{proof}
 See Appendix \ref{a42}.
\end{proof}
\section{A General Model with Bounded Rationality}
In the previous sections, we only consider the classical case where all agents are assumed to be perfectly rational. However, this assumption is rather strong for many economic applications where agents are shown to be bounded rational, meaning that they act on the neighborhood of best responses, but not always the best responses. This naturally gives rise to the concept of randomized policies in the reinforcement learning literature or relaxed control problems in stochastic control theories. Indeed, by introducing it together with the Shannon entropy into the objective optimization function, we are able to describe the system where agents are bounded rational (see, e.g., Goeree et al., 2020).  \\

Suppose that the representative agent of group $j$ adopts a saving strategy $\pi_c^j=\{\pi_{c,t}^j\in\mathbb P(C),t\in[0,T]\}$, where the saving action space $C$ is a closed subset of a Euclidean space and $\mathbb P(C)$ is the set of density functions of probability measures on $C$ that is absolutely continuous with respect to the Lebesgue measure, on the admissible strategy set $\mathcal A_c$. In addition, we introduce a new situation in which the risk exposure of the market $\phi_j$ can be "chosen" up to a boundary $\tilde{\phi_j}$. That is, now the representative agent of the group $j$ can control a risk strategy $\pi_\phi^j=\{\pi_{\phi,t}^j\in\mathbb P(\varphi),t\in[0,T]\}$, where the risk action space $\varphi$ is a closed subset of a Euclidean space and $\mathbb P(\varphi)$ is the set of density functions of probability measures on $\varphi$ that are absolutely continuous with respect to the Lebesgue measure, on the admissible strategy set $\mathcal A_\phi$. Now, we can rewrite the controlled wealth process of $\eqref{107}$ for a representative agent in group $j$ as follows:
\begin{align}
    dw_{j,t}=\bigg(\int_C b_j^w(.)\pi_{c,t}^j(c)dc\bigg)dt+q_t\sigma_{j,k}dB_t^j+\bigg(\sqrt{\int_\varphi v_j(.)\pi_{\phi,t}^j(\phi)d\phi}\bigg)dW_t,
\end{align}
for $t\geq 0$.\\

The representative agent in group $j$ will maximize their payoff function as follows:
\begin{equation}
\begin{split}
    J_j=\sup_{\pi_{c,t}^j\in \mathcal A_c, \pi_{\phi,t}^j\in \mathcal A_\phi}\mathbb E\bigg[\int_0^T e^{-\rho_j t}\bigg(\int_CU(c_{j,t})\pi_{c,t}^j(c)dc-\int_\varphi\frac{\phi_j^2}{2}\pi_{\phi,t}^j(\phi)d\phi\\
    -\lambda_c\int_C\pi_{c,t}^j(c)\ln \pi_{c,t}^j(c)dc -\lambda_\phi\int_\varphi\pi_{\phi,t}^j(\phi)\ln \pi_{\phi,t}^j(\phi)d\phi\bigg)dt +w_{j,T}\bigg],
\end{split}
\end{equation}

\begin{proposition}
    The optimal consumption strategy of agents in group $j$ is given by:
    \begin{align}
    \widehat{\pi}_{c,t}^j(c)=\exp\bigg[\frac{(c_{j,t})^{-\gamma}-y_j^w\cdot e^{\rho_jt}}{\lambda_c}-1\bigg],  
\end{align}
where $y^w_j=e^{-r(T-t)}$ and $0\leq t\leq T$.   
\end{proposition}
\begin{proof}
    See Appendix \ref{a51}
\end{proof}

\section{Conclusion}
In conclusion, this work offers a novel perspective on the persistent sluggishness of post-crisis economic recoveries. We depart from the limitations of traditional models by incorporating financial frictions and agent heterogeneity within a novel mean field game (MFG) framework. This framework sheds light on the crucial role of financial frictions, which act as impediments to optimal capital allocation during downturns, thus hindering the flow of resources to high-growth sectors and disrupting the essential restructuring process for economic revival.\\

Furthermore, the MFG framework elegantly captures the distinct behavior of expert and household agents. Our analysis underscores the significance of agent heterogeneity, as expert agents with superior capital management prowess accumulate capital during boom periods and strategically deleverage during downturns. In contrast, household agents exhibit a rational response by adjusting their holdings based on the observed behavior of experts. The introduction of a mean-field interaction within the framework further enriches our understanding. This interaction captures the dynamic interplay between agents within each group, where strategic reactions to one another amplify the overall dynamics. This mechanism explains the observed cautious behavior of households during crises and the inherently slower recovery process compared to models with homogeneous agents. Using the MFG framework, we contribute a powerful and tractable tool to analyze complex economic systems populated by a multitude of heterogeneous agents.  This research deepens our understanding of the intricate interplay between financial frictions, agent heterogeneity, and mean field interaction in shaping economic recovery dynamics.\\

These findings have significant policy implications, highlighting the critical need to address financial frictions and promote efficient capital allocation across various sectors. Ultimately, such efforts pave the way for the creation of more robust and resilient economic systems. Looking ahead, future research efforts could fruitfully explore the impact of alternative policy interventions and the potential to incorporate additional agent types with varying risk tolerances and information access into the MFG framework.

\appendix

\section{Proof of Proposition 2.1}\label{a21}
First, we construct the adjoint FBSDEs and find the corresponding reduced Hamiltonian: 
\begin{align}
    H^e:[0,T]\times\mathbb R\times\mathbb R^2\times\mathbb R \times\mathbb R \times\mathbb R^2 \times\mathbb R^2 
\end{align}
\begin{align}
     H^h:[0,T]\times\mathbb R\times\mathbb R^2\times\mathbb R \times\mathbb R \times\mathbb R^2 \times\mathbb R^2 
\end{align}
associated to problem (\ref{14}) and (\ref{17}), which are defined as follows:
\begin{align*}
    H^e(t,w_{e,t},y,c_{e,t},q_t,\hat{k_t^.},k_{.,t})=b_e(w_{e,t},q_t,\widehat{k}_t^e,k_{e,t},c_{e,t})\cdot y_e+e^{-\rho t}u(c_{e,t})\\
    +b_h(w_{h,t},q_t,\widehat{k_t^h},k_{h,t},c_{h,t}^i)\cdot y_h,
\end{align*}
\begin{align*}
    H^h(t,w_{h,t},y,c_{h,t},q_t,\hat{k_t^.},k_{.,t})=b_h(w_{h,t},q_t,\widehat{k}_t^h,k_{h,t},c_{h,t})\cdot y_h+e^{-\rho t}u(c_{h,t})\\
    +b_e(w_{e,t},q_t,\hat{k_t^e},k_{e,t},c_{e,t})\cdot y_e.
\end{align*}
The sufficient and necessary conditions are set similarly to the standard problems referred to the Section 5.2 of Oksendal \& Sulem (2018) so that we can apply the Pontryagin Stochastic Maximum Principles. \\

Minimizing the Hamiltonian with respect to $c$, we have:
\begin{align}
    \frac{\partial H^e(t,w_{e,t},y,c_{e,t},q_t,\hat{k_t^.},k_{.,t})}{\partial c}=-y_{e}+e^{-\rho t}(c_{e,t})^{-\gamma},
\end{align}
\begin{align}
    \frac{\partial H^h(t,w_{h,t},y,c_{h,t},q_t,\hat{k_t^.},k_{.,t})}{\partial c}=-y_{h}+e^{-rt}(c_{h,t})^{-\gamma}.
\end{align}

Hence, we get the optimal control:
\begin{align}
    \widehat{c}_{e,t}=e^{-\frac{\rho t}{\gamma}}(y_e)^{-\frac{1}{\gamma}},
\end{align}
\begin{align}
    \widehat{c}_{h,t}=e^{-\frac{r t}{\gamma}}(y_h)^{-\frac{1}{\gamma}}.
\end{align}
The system of backward equations satisfy:
\begin{equation}
\begin{split}
     dY_{e,t}^w&=-\partial_{w_{e,t}}H^{\text{full}}(t,w_{e,t},Y_{e,t}^w,z,\widehat{c}_{e,t})dt+\sum_{i\in\{e,h\}}Z_t^{Bwei}dW_t^{B,i}+\sum_{i\in\{e,h\}}Z_t^{wei}dB_t^{i}\\
    &=-rY_{e,t}dt+\sum_{i\in\{e,h\}}Z_t^{Bwei}dW_t^{B,i}+\sum_{i\in\{e,h\}}Z_t^{wei}dB_t^{i},
\end{split}
\end{equation}

\begin{equation}
\begin{split}
   dY_{h,t}^w&=-\partial_{w_{h,t}}H^{\text{full}}(t,w_{h,t},Y_{h,t}^w,z,\widehat{c}_{h,t})dt+\sum_{i\in\{e,h\}}Z_t^{Bwhi}dW_t^{B,i}+\sum_{i\in\{e,h\}}Z_t^{whi}dB_t^{i}\\
    &=-rY_{h,t}dt+\sum_{i\in\{e,h\}}Z_t^{Bwhi}dW_t^{B,i}+\sum_{i\in\{e,h\}}Z_t^{whi}dB_t^{i},
\end{split}
\end{equation}

\begin{equation}
\begin{split}
dY_{e,t}^k&=-\partial_{k_{e,t}}H^{\text{full}}(t,k_{e,t},Y_{e,t}^k,z,\widehat{c}_{e,t})dt+\sum_{i\in\{e,h\}}Z_t^{kei}dB_t^{i}\\
    &=\bigg[q_t|\Phi(\iota_t)-\delta_e|-\mu_t^qq_t+q_tr\bigg]Y_{e,t}dt+|\Phi(\iota_t)-\delta_e|\cdot Y_{e,t}^kdt-\phi q_t\sigma_t^qZ_{e,t}dt\\
    & \hspace{1cm} {} 
    +\sum_{i\in\{e,h\}}Z_t^{kei}dB_t^{i},
\end{split}
\end{equation}

\begin{equation}
\begin{split}
    dY_{h,t}^k&=-\partial_{k_{h,t}}H^{\text{full}}(t,k_{h,t},Y_{h,t}^k,z,\widehat{c}_{h,t})dt+\sum_{i\in\{e,h\}}Z_t^{kei}dB_t^{i}\\
    &=\bigg[q_t(1-\delta_h)-\mu_t^qq_t+q_tr\bigg]Y_{h,t}dt+(1-\delta_h)Y_{h,t}^k-(1-\phi )q_t\sigma_t^qZ_{h,t}dt\\
    & \hspace{1cm} {} 
    +\sum_{i\in\{e,h\}}Z_t^{kei}dB_t^{i},
\end{split}
\end{equation}
with terminal condition $Y_{e,T}=-1$ and $Y_{h,T}=-1$.
We omit the non-diagonal adjoint processes which can be treated analogously. \\

The FBSDE \eqref{eq: FBSDE1} for the capital processes derived from the Pontryagin stochastic maximum principle (see e.g Oksendal \& Sulem (2018)) reads:
\begin{equation}  \label{eq: FBSDE1}
   \begin{split}
   dk_{e,t}&=[M_{e,t}-|\Phi(\iota_t)-\delta_e|\cdot k_{e,t}]dt+\sigma_{e,k}dB^e_t,\\
   dk_{h,t}& =(M_{h,t}-(1-\delta_h)k_{h,t})dt+\sigma_{h,k}dB^h_t,\\
   dY_{e,t}^k&=\bigg[q_t|\Phi(\iota_t)-\delta_e|-\mu_t^qq_t+q_tr\bigg]Y_{e,t}dt+|\Phi(\iota_t)-\delta_e|\cdot Y_{e,t}^kdt-\phi q_t\sigma_t^qZ_{e,t}dt\\
    & \hspace{1cm} {} 
    +\sum_{i\in\{e,h\}}Z_t^{kei}dB_t^{i},\\
    dY_{h,t}^k&=\bigg[q_t(1-\delta_h)-\mu_t^qq_t+q_tr\bigg]Y_{h,t}dt+(1-\delta_h)Y_{h,t}^k-(1-\phi )q_t\sigma_t^qZ_{h,t}dt\\
    & \hspace{1cm} {} 
    +\sum_{i\in\{e,h\}}Z_t^{kei}dB_t^{i},
   \end{split}
\end{equation}
with terminal conditions $ Y_{e,T}^k=Y_{h,T}^k =0$.\\

Similarly, the FBSDE \eqref{eq: FBSDE2} for the wealth processes derived from the Pontryagin stochastic maximum principle (see e.g Oksendal \& Sulem (2018)) reads:
\begin{equation}  \label{eq: FBSDE2}
   \begin{split}
   dw_{e,t}&=\bigg[rw_{e,t}+q_t(M_{e,t}-|\Phi(\iota_t)-\delta_e|\cdot k_{e,t}) \\
        & \hspace{1cm} {} +k_{e,t}\mu_t^qq_t+p_1\sigma_t^q\sigma_{e,k}-k_{e,t}q_tr-e^{-\frac{\rho t}{\gamma}}(Y_{e,t}^w)^{-\frac{1}{\gamma}}\bigg]dt 
         +v_edW_t+q_t\sigma_{e,t}dB^e_t,\\
    dw_{h,t}&=\bigg[rw_{h,t}+q_t(M_{h,t}-(1-\delta_h)k_{h,t})+k_{h,t}\mu_t^qq_t+p_2\sigma_t^q\sigma_{h,k} \\
        & \hspace{1cm} {} 
        -k_{h,t}q_tr-e^{-\frac{r t}{\gamma}}(Y_{h,t}^w)^{-\frac{1}{\gamma}}\bigg]dt 
        +v_hdW_t+q_t\sigma_{h,t}dB^h_t,\\
    dY_{e,t}^w&=-rY_{e,t}dt+\sum_{i\in\{e,h\}}Z_t^{Bwei}dW_t^{B,i}+\sum_{i\in\{e,h\}}Z_t^{wei}dB_t^{i},\\
    dY_{h,t}^w&=-rY_{h,t}dt+\sum_{i\in\{e,h\}}Z_t^{Bwhi}dW_t^{B,i}+\sum_{i\in\{e,h\}}Z_t^{whi}dB_t^{i},
   \end{split}
\end{equation}
with terminal condition $Y_{e,T}=-1$ and $Y_{h,T}=-1$.\\

It is worth noting that the following backward equations are negligible as they are decoupled from other FBSDEs and independent of the controls:
\begin{equation}
    \begin{split}
       dY_{e,t}^k&=\bigg[q_t|\Phi(\iota_t)-\delta_e|-\mu_t^qq_t+q_tr\bigg]Y_{e,t}dt+|\Phi(\iota_t)-\delta_e|\cdot Y_{e,t}^kdt-\phi q_t\sigma_t^qZ_{e,t}dt\\
     &\hspace{1cm} {} 
    +\sum_{i\in\{e,h\}}Z_t^{kei}dB_t^{i},\\
    dY_{h,t}^k&=\bigg[q_t(1-\delta_h)-\mu_t^qq_t+q_tr\bigg]Y_{h,t}dt+(1-\delta_h)Y_{h,t}^k-(1-\phi )q_t\sigma_t^qZ_{h,t}dt\\
     &\hspace{1cm} {} 
    +\sum_{i\in\{e,h\}}Z_t^{kei}dB_t^{i},\\
        Y_{e,T}^k&=0,\\
        Y_{h,T}^k&=0.
    \end{split}
\end{equation}
Also, the other backward equations are not negligible but may be separated from the forward equation as the BSDE is constant with respect to the wealth process:
\begin{equation}
    \begin{split}
     dY_{e,t}^w&=-rY_{e,t}dt+\sum_{i\in\{e,h\}}Z_t^{Bwei}dW_t^{B,i}+\sum_{i\in\{e,h\}}Z_t^{wei}dB_t^{i},\\
    dY_{h,t}^w&=-rY_{h,t}dt+\sum_{i\in\{e,h\}}Z_t^{Bwhi}dW_t^{B,i}+\sum_{i\in\{e,h\}}Z_t^{whi}dB_t^{i},\\
        Y_{e,T}^w&=-1,\\
        Y_{h,T}^w&=-1.
    \end{split}
\end{equation}
Their deterministic solution (i.e. by setting $Z=0$) is achieved by solving the following backward ordinary differential equation:  
\begin{equation}
    \begin{split}
         dY_{e,t}&=-rY_{e,t}dt,\\
        dY_{h,t}&=-rY_{h,t}dt,\\
        Y_{e,T}&=-1,\\
        Y_{h,T}&=-1.
    \end{split}
\end{equation}
The solution for this type of ODE is elementary and is given explicitly in the form:
\begin{equation}
    \begin{split}
        Y_{e,t}&=\exp\bigg\{-\int_t^Trds\bigg\}=e^{-r(T-t)},\\
        Y_{h,t}&=\exp\bigg\{-\int_t^Trds\bigg\}=e^{-r(T-t)}.
    \end{split}
\end{equation}

\section{Proof of Theorem 2.2}\label{a22}
 First, we solve the fixed point problem by taking the expectation of the forward equations of the capital processes:
\begin{equation}
    \begin{split}
        dm_{e,t}=\bigg[M_{e,t}-|\Phi(\iota_t)-\delta_e|\cdot m_{e,t}\bigg]dt+\sigma_{e,k}dB^e_t,\\
        dm_{h,t}=\bigg[M_{h,t}-(1-\delta_h)m_{h,t}\bigg]dt+\sigma_{h,k}dB^h_t.
    \end{split}\label{66}
\end{equation}

 To solve the system of equations \eqref{66}, we set $$\lim_{N_e,N_h\to\infty}\frac{N_e}{N_e+N_h}:=\alpha_e,$$ and $$\lim_{N_e,N_h\to\infty}\frac{N_h}{N_e+N_h}:=\alpha_h,$$ such that $$\alpha_e+\alpha_h=1,$$ then we have:
 
\begin{equation}
    \begin{split}
        M_{e,t}&=(1-\lambda_e)m_{e,t}+\lambda_e(\alpha_em_{e,t}+\alpha_hm_{h,t}),\\
        M_{h,t}&=(1-\lambda_h)m_{h,t}+\lambda_h(\alpha_em_{e,t}+\alpha_hm_{h,t}).
    \end{split}\label{67}
\end{equation}

Rearranging terms, we have:

\begin{equation}
    \begin{split}
        M_{e,t}&=(1-\lambda_e+\alpha_e\lambda_e)m_{e,t}+\lambda_e\alpha_hm_{h,t},\\
        M_{h,t}&=\lambda_h\alpha_em_{e,t}+(1-\lambda_h+\alpha_h\lambda_h)m_{h,t}.
    \end{split}
\end{equation}

Substituting Equations \eqref{67} into Equations \eqref{66}, we have:

\begin{equation}
    \begin{split}
         dm_{e,t}&=\bigg[(1-\lambda_e)m_{e,t}+\lambda_e(\alpha_em_{e,t}+\alpha_hm_{h,t})-|\Phi(\iota_t)-\delta_e|\cdot m_{e,t}\bigg]dt+\sigma_{e,k}dB^e_t,\\
        dm_{h,t}&=\bigg[(1-\lambda_h)m_{h,t}+\lambda_h(\alpha_em_{e,t}+\alpha_hm_{h,t})-(1-\delta_h)m_{h,t}\bigg]dt+\sigma_{h,k}dB^h_t.
    \end{split}
\end{equation}

Rearranging terms, we have the following.

\begin{equation}
    \begin{split}
         dm_{e,t}&=\bigg[(1-\lambda_e+\lambda_e\alpha_e-|\Phi(\iota_t)-\delta_e|)m_{e,t}+\lambda_e\alpha_hm_{h,t}\bigg]dt+\sigma_{e,k}dB^e_t,\\
        dm_{h,t}&=\bigg[\lambda_h\alpha_em_{e,t}+(-\lambda_h+\lambda_h\alpha_h+\delta_h)m_{h,t}\bigg]dt+\sigma_{h,k}dB^h_t.
    \end{split}\label{69}
\end{equation}
If we set $m_t:=\begin{pmatrix}m_{e,t}\\m_{h,t}\end{pmatrix}$, then we can rewrite system \eqref{69} into a succinct form as follows:
\begin{align}
    dm_t=A_tm_tdt+LdB^l_t,
\end{align}
where 
\begin{align*}
    A_t:=\begin{pmatrix}
    1-\lambda_e+\lambda_e\alpha_e-|\Phi(\iota_t)-\delta_e| & \lambda_e\alpha_h\\
    \lambda_h\alpha_e & -\lambda_h+\lambda_h\alpha_h+\delta_h
    \end{pmatrix},
\end{align*}
 $$L=\begin{pmatrix}\sigma_{e,k}&0 \\ 0&\sigma_{h,k} \end{pmatrix},$$
 and 
\begin{align}
    B^l_t=\begin{pmatrix} B^e_t \\ B^h_t \end{pmatrix}.
\end{align}

This linear SDE has a unique analytical solution (see e.g Karatzas \& Shreve (2014)), which is represented in the following form:
\begin{align}
    m_t=v_t\bigg[m_0+\int_{0}^tv^{-1}_sL dB^l_s\bigg]; \quad 0\leq t\leq T,
\end{align}
where $v_t=e^{\int_0^\infty A_sds}$ is the fundamental solution of the following homogeneous equation:
\begin{align*}
    dv_t=A_tv_tdt; \quad v_0=I,
\end{align*}
where $I$ is the $2\times 2$ identity matrix.

\section{Proof of Proposition 3.1}\label{a31}
First, we construct the adjoint FBSDEs and find the corresponding reduced Hamiltonian: 
\begin{align}
    H^e:[0,T]\times\mathbb R\times\mathbb R^2\times\mathbb R \times\mathbb R \times\mathbb R^2 \times\mathbb R^2 
\end{align}
\begin{align}
     H^h:[0,T]\times\mathbb R\times\mathbb R^2\times\mathbb R \times\mathbb R \times\mathbb R^2 \times\mathbb R^2 
\end{align}
associated to problem (\ref{14}) and (\ref{17}), which are defined as follows:
\begin{align}
    H^e(t,w_{e,t},y,c_{e,t},q_t,\hat{k_t^.},k_{.,t})=b_e^w(w_{e,t},q_t,\hat{k_t^e},k_{e,t},c_{e,t})\cdot y_e+e^{-\rho t}u(c_{e,t})\\
    +b_h^w(w_{h,t},q_t,\hat{k_t^h},k_{h,t},c_{h,t}^i)\cdot y_h,
\end{align}
\begin{align}
    H^h(t,w_{h,t},y,c_{h,t},q_t,\hat{k_t^.},k_{.,t})=b_h^w(w_{h,t},q_t,\hat{k_t^h},k_{h,t},c_{h,t})\cdot y_h+e^{-\rho t}u(c_{h,t})\\
    +b_e^w(w_{e,t},q_t,\hat{k_t^e},k_{e,t},c_{e,t})\cdot y_e.
\end{align}
The sufficient and necessary conditions are set similarly to the standard problems referred to the Section 5.2 of Oksendal \& Sulem (2018) so that we can apply the Pontryagin Stochastic Maximum Principles. \\

Minimizing the Hamiltonian for consumption control $c$, we have:
\begin{align}
    \frac{\partial H^e(t,w_{e,t},y,c_{e,t},q_t,\hat{k_t^.},k_{.,t})}{\partial c}=-y_{e}+e^{-\rho t}(c_{e,t})^{-\gamma},
\end{align}
\begin{align}
    \frac{\partial H^h(t,w_{h,t},y,c_{h,t},q_t,\hat{k_t^.},k_{.,t})}{\partial c}=-y_{h}+e^{-rt}(c_{h,t})^{-\gamma}.
\end{align}

Hence, we get the optimal control:
\begin{align}
    \widehat{c_{e,t}}=e^{-\frac{\rho t}{\gamma}}(y_e)^{-\frac{1}{\gamma}},
\end{align}
\begin{align}
    \widehat{c_{h,t}}=e^{-\frac{r t}{\gamma}}(y_h)^{-\frac{1}{\gamma}}.
\end{align}

The system of backward equations satisfies:
\begin{equation}
\begin{split}
    dY_{e,t}^w&=-\partial_{w_{e,t}}H^{\text{full}}(t,w_{e,t},Y_{e,t}^w,z,\widehat{c_{e,t}})dt+\sum_{i\in\{e,h\}}Z_t^{Bwei}dW_t^{B,i}+\sum_{i\in\{e,h\}}Z_t^{wei}dB_t^{i}\\
    &=-rY_{e,t}dt+\sum_{i\in\{e,h\}}Z_t^{Bwei}dW_t^{B,i}+\sum_{i\in\{e,h\}}Z_t^{wei}dB_t^{i},
\end{split}
\end{equation}
\begin{equation}
\begin{split}
    dY_{h,t}^w&=-\partial_{w_{h,t}}H^{\text{full}}(t,w_{h,t},Y_{h,t}^w,z,\widehat{c_{h,t}})dt+\sum_{i\in\{e,h\}}Z_t^{Bwhi}dW_t^{B,i}+\sum_{i\in\{e,h\}}Z_t^{whi}dB_t^{i}\\
    &=-rY_{h,t}dt+\sum_{i\in\{e,h\}}Z_t^{Bwhi}dW_t^{B,i}+\sum_{i\in\{e,h\}}Z_t^{whi}dB_t^{i},
\end{split}
\end{equation}

\begin{equation}
\begin{split}
    dY_{e,t}^k&=-\partial_{k_{e,t}}H^{\text{full}}(t,k_{e,t},Y_{e,t}^k,z,\widehat{c_{e,t}})dt+\sum_{i\in\{e,h\}}Z_t^{Bkei}dW_t^{B,i}+\sum_{i\in\{e,h\}}Z_t^{kei}dB_t^{i}\\
    &=\bigg[q_t|\Phi(\iota_t)-\delta_e|-\mu_t^qq_t+q_tr\bigg]Y_{e,t}dt+|\Phi(\iota_t)-\delta_e|\cdot Y_{e,t}^kdt-\phi q_t\sigma_t^qZ_{e,t}dt\\
    & \hspace{1cm} {} 
    +\sum_{i\in\{e,h\}}Z_t^{Bkei}dW_t^{B,i}+\sum_{i\in\{e,h\}}Z_t^{kei}dB_t^{i},
\end{split}
\end{equation}

\begin{equation}
\begin{split}
    dY_{h,t}^k&=-\partial_{k_{h,t}}H^{\text{full}}(t,k_{h,t},Y_{h,t}^k,z,\widehat{c_{h,t}})dt+\sum_{i\in\{e,h\}}Z_t^{Bkhi}dW_t^{B,i}+\sum_{i\in\{e,h\}}Z_t^{khi}dB_t^{i}\\
    &=\bigg[q_t(1-\delta_h)-\mu_t^qq_t+q_tr\bigg]Y_{h,t}dt+(1-\delta_h)Y_{h,t}^kdt-(1-\phi)q_t\sigma_t^qZ_{h,t}dt\\
     & \hspace{1cm} {} +\sum_{i\in\{e,h\}}Z_t^{Bkhi}dW_t^{B,i}+\sum_{i\in\{e,h\}}Z_t^{khi}dB_t^{i},
\end{split}
\end{equation}
with terminal condition $Y_{e,T}=-1$ and $Y_{h,T}=-1$.
We omit the non-diagonal adjoint processes which can be treated analogously.\\ 

The FBSDE \eqref{eq: FBSDE3} for the capital processes derived from the Pontryagin stochastic maximum principle (see e.g Oksendal \& Sulem (2018)) reads:
\begin{equation}  \label{eq: FBSDE3}
   \begin{split}
    dk_{e,t}&=[M_{e,t}-|\Phi(\iota_t)-\delta_e|\cdot k_{e,t}]dt+\sigma_{e,k}dB^e_t+\sigma dW_t,\\
     dk_{h,t}&=(M_{h,t}-(1-\delta_h)k_{h,t})dt+\sigma_{h,k}dB^h_t+\sigma dW_t,\\
     dY_{e,t}^k&=\bigg[q_t|\Phi(\iota_t)-\delta_e|-\mu_t^qq_t+q_tr\bigg]Y_{e,t}dt+|\Phi(\iota_t)-\delta_e|\cdot Y_{e,t}^kdt-\phi q_t\sigma_t^qZ_{e,t}dt\\
    & \hspace{1cm} {} 
    +\sum_{i\in\{e,h\}}Z_t^{Bkei}dW_t^{B,i}+\sum_{i\in\{e,h\}}Z_t^{kei}dB_t^{i},\\
     dY_{h,t}^k&=\bigg[q_t(1-\delta_h)-\mu_t^qq_t+q_tr\bigg]Y_{h,t}dt+(1-\delta_h)Y_{h,t}^kdt-(1-\phi)q_t\sigma_t^qZ_{h,t}dt\\
     & \hspace{1cm} {} +\sum_{i\in\{e,h\}}Z_t^{Bkhi}dW_t^{B,i}+\sum_{i\in\{e,h\}}Z_t^{khi}dB_t^{i},
    \end{split}
\end{equation}
with terminal condition $Y_{e,T}^k=Y_{h,T}^k=0$.\\

Similarly, the FBSDE \eqref{eq: FBSDE4} for the wealth processes derived from the Pontryagin stochastic maximum principle (see e.g Oksendal \& Sulem (2018)) reads:
\begin{equation}  \label{eq: FBSDE4}
   \begin{split}
   dw_{e,t}&=\bigg[rw_{e,t}+q_t(M_{e,t}-(|\Phi(\iota_t)-\delta_e|\cdot k_{e,t})+k_{e,t}\mu_t^qq_t+\sigma_t^q\sigma q_t\\
        & \hspace{1cm} {}
        +p_1\sigma_t^q\sigma_{e,k}q_t-k_{e,t}q_tr-e^{-\frac{\rho t}{\gamma}}(Y_{e,t}^w)^{-\frac{1}{\gamma}}\bigg]dt+q_t\sigma_{e,k}dB_t^e+v_edW_t,\\
    dw_{h,t}&=\bigg[rw_{h,t}+q_t(M_{h,t}-(1-\delta_h)k_{h,t})+k_{h,t}\mu_t^qq_t+\sigma_t^q\sigma q_t\\
        & \hspace{1cm} {}
        +p_2\sigma_t^q\sigma_{h,k}q_t-k_{h,t}q_tr-e^{-\frac{r t}{\gamma}}(Y_{h,t}^w)^{-\frac{1}{\gamma}}\bigg]+q_t\sigma_{h,k}dB_t^h+v_hdW_t,\\
    dY_{e,t}^w&= -rY_{e,t}dt+\sum_{i\in\{e,h\}}Z_t^{Bwei}dW_t^{B,i}+\sum_{i\in\{e,h\}}Z_t^{wei}dB_t^{i},\\
    dY_{h,t}^w&=-rY_{h,t}dt+\sum_{i\in\{e,h\}}Z_t^{Bwhi}dW_t^{B,i}+\sum_{i\in\{e,h\}}Z_t^{whi}dB_t^{i},
   \end{split}
\end{equation}
with terminal condition $Y_{e,T}^w=Y_{h,T}^w=-1$.\\

Again, we can decouple and disregard the BSDEs of the capital processes and consider a separate system of the wealth processes BSDEs as in the case without aggregate shocks. Their deterministic solution is achieved by solving the following backward ordinary differential equation:  
\begin{equation}
    \begin{split}
         dY_{e,t}&=-rY_{e,t}dt,\\
        dY_{h,t}&=-rY_{h,t}dt,\\
        Y_{e,T}&=-1,\\
        Y_{h,T}&=-1.
    \end{split}
\end{equation}

The solution for this type of ODE is elementary and is given explicitly in the form:
\begin{equation}
    \begin{split}
        Y_{e,t}=\exp\bigg\{-\int_t^Trds\bigg\}=e^{-r(T-t)},\\
        Y_{h,t}=\exp\bigg\{-\int_t^Trds\bigg\}=e^{-r(T-t)}.
    \end{split}
\end{equation}

\section{Proof of Theorem 3.2}\label{a32}
First, we solve the fixed point problem by taking the expectation of the forward equations of the capital processes:
\begin{equation}
    \begin{split}
        dm_{e,t}=\bigg[M_{e,t}-|\Phi(\iota_t)-\delta_e|\cdot m_{e,t}\bigg]dt+\sigma_{e,k}dB^e_t+\sigma dW_t,\\
        dm_{h,t}=\bigg[M_{h,t}-(1-\delta_h)m_{h,t}\bigg]dt+\sigma_{h,k}dB^h_t+\sigma dW_t.
    \end{split}\label{86}
\end{equation}

 We have:
\begin{equation}
    \begin{split}
        M_{e,t}=(1-\lambda_e)m_{e,t}+\lambda_e(\alpha_em_{e,t}+\alpha_hm_{h,t}),\\
        M_{h,t}=(1-\lambda_h)m_{h,t}+\lambda_h(\alpha_em_{e,t}+\alpha_hm_{h,t}).
    \end{split}\label{87}
\end{equation}

Rearranging terms, we have:

\begin{equation}
    \begin{split}
        M_{e,t}=(1-\lambda_e+\alpha_e\lambda_e)m_{e,t}+\lambda_e\alpha_hm_{h,t},\\
        M_{h,t}=\lambda_h\alpha_em_{e,t}+(1-\lambda_h+\alpha_h\lambda_h)m_{h,t}.
    \end{split}
\end{equation}

Substitute Equations \eqref{87} into Equations \eqref{86}, we have:

\begin{equation}
    \begin{split}
         dm_{e,t}=\bigg[(1-\lambda_e)m_{e,t}+\lambda_e(\alpha_em_{e,t}+\alpha_hm_{h,t})-|\Phi(\iota_t)-\delta_e|\cdot m_{e,t}\bigg]dt+\sigma_{e,k}dB^e_t+\sigma dW_t,\\
        dm_{h,t}=\bigg[(1-\lambda_h)m_{h,t}+\lambda_h(\alpha_em_{e,t}+\alpha_hm_{h,t})-(1-\delta_h)m_{h,t}\bigg]dt+\sigma_{h,k}dB^h_t+\sigma dW_t.
    \end{split}
\end{equation}

Rearranging terms, we have:

\begin{equation}
    \begin{split}
         dm_{e,t}=\bigg[(1-\lambda_e+\lambda_e\alpha_e-|\Phi(\iota_t)-\delta_e|)m_{e,t}+\lambda_e\alpha_hm_{h,t}\bigg]dt+\sigma_{e,k}dB^e_t+\sigma dW_t,\\
        dm_{h,t}=\bigg[\lambda_h\alpha_em_{e,t}+(-\lambda_h+\lambda_h\alpha_h+\delta_h)m_{h,t}\bigg]dt+\sigma_{h,k}dB^h_t+\sigma dW_t.
    \end{split}\label{88}
\end{equation}

We can rewrite system \eqref{88} into a succinct form as follows:
\begin{align}
    dm_t=A_tm_tdt+LdB^l_t+SdW^l_t,
\end{align}
where 
\begin{align*}
    A_t:=\begin{pmatrix}
    1-\lambda_e+\lambda_e\alpha_e-|\Phi(\iota_t)-\delta_e| & \lambda_e\alpha_h\\
    \lambda_h\alpha_e & -\lambda_h+\lambda_h\alpha_h+\delta_h
    \end{pmatrix},
\end{align*}
 $$S=\begin{pmatrix}\sigma&0 \\ 0&\sigma \end{pmatrix}$$
\begin{align}
    L=\begin{pmatrix}\sigma_{e,k}&0 \\ 0&\sigma_{h,k} \end{pmatrix},
\end{align}
and 
\begin{align}
    B^l_t&=\begin{pmatrix} B^e_t \\ B^h_t \end{pmatrix},\\
    W^l_t&=\begin{pmatrix} W_t \\ W_t \end{pmatrix}.
\end{align}
 
This linear SDE has a unique analytical solution (see e.g Karatzas \& Shreve (2014)), which is represented in the following form:
\begin{align}
    m_t=v_t\bigg[m_0+\int_{0}^tv^{-1}_sL dB^l_s+\int_{0}^tv^{-1}_sS dW^l_s\bigg]; \quad 0\leq t\leq T,
\end{align}
where $v_t=e^{\int_0^\infty A_sds}$ is the fundamental solution of the following homogeneous equation:
\begin{align*}
    dv_t=A_tv_tdt; \quad v_0=I,
\end{align*}
where $I$ is the $2\times 2$ identity matrix.

\section{Proof of Proposition 4.1}\label{a41}
First, we construct the adjoint FBSDEs and find the corresponding reduced Hamiltonian: 
\begin{align}
    H^j:[0,T]\times\mathbb R^n\times\mathbb R^n\times\mathbb R^n\times\mathbb R \times\mathbb R \times\mathbb R^n \times\mathbb R^n 
\end{align}

associated to problem (\ref{14}) and (\ref{17}), which are defined as follows:
\begin{equation}
    \begin{split}
H^j(t,w_{.,t},y^w,y^k,c_{j,t},q_t,\hat{k_t^.},k_{.,t})=\sum_{s=1}^n\big[b_s^w(w_{s,t},q_t,\hat{k_t^s},k_{s,t},c_{s,t})\cdot y^w_s
+u_j^k\cdot y^k_s\big]
+e^{-\rho t}u(c_{j,t}).
    \end{split}
\end{equation}

The sufficient and necessary conditions are set similarly to the standard problems referred to the Section 5.2 of Oksendal \& Sulem (2018) so that we can apply the Pontryagin Stochastic Maximum Principles. \\

Minimizing the Hamiltonian for consumption control $c_{j}$, we have:
\begin{align}
    \frac{\partial H^j(t,w_{j,t},y^w,y^k,c_{j,t},q_t,\hat{k_t^.},k_{.,t})}{\partial c_j}=-y^w_{j}+e^{-\rho t}(c_{j,t})^{-\gamma}.
\end{align}

Hence, we get the optimal control:
\begin{align}
    \widehat{c}_{j,t}=e^{-\frac{\rho t}{\gamma}}(y^w_j)^{-\frac{1}{\gamma}}.
\end{align}

The system of backward equations satisfies:
\begin{equation}
\begin{split}
    dY_{j,t}^w&=-\partial_{w_{j,t}}H^j_{\text{full}}( \widehat{c}_{j,t})dt+\sum_{i=1}^nZ_t^{Bwji}dW_t^{B,i}+\sum_{i=1}^nZ_t^{wji}dB_t^{i}\\
    &=-rY^w_{j,t}dt+\sum_{i=1}^nZ_t^{Bwji}dW_t^{B,i}+\sum_{i=1}^nZ_t^{wji}dB_t^{i},\\
    dY_{j,t}^k&=-\partial_{k_{j,t}}H^j_{\text{full}}( \widehat{c}_{j,t})dt+\sum_{i=1}^nZ_t^{Bkji}dW_t^{B,i}+\sum_{i=1}^nZ_t^{kji}dB_t^{i}\\
    &=\bigg[q_t|\Phi(\iota_{j,t})-\delta_j|-\mu_t^qq_t+q_tr\bigg]Y_{j,t}^wdt+|\Phi(\iota_{j,t})-\delta_j|\cdot Y_{j,t}^kdt-\phi_j q_t\sigma_t^qZ_{j,t}dt\\
    & \hspace{1cm} {} 
    +\sum_{i=1}^nZ_t^{Bkji}dW_t^{B,i}+\sum_{i=1}^nZ_t^{kji}dB_t^{i},
\end{split}
\end{equation}
with terminal condition $Y_{j,T}=-1$ and $Y^k_{j,T}=0$ for all $j$. We omit the non-diagonal adjoint processes which can be treated analogously.\\ 

The FBSDE \eqref{eq: FBSDE8} for the capital processes derived from the Pontryagin stochastic maximum principle (see e.g Oksendal \& Sulem (2018)) reads:
\begin{equation}  \label{eq: FBSDE8}
   \begin{split}
    dk_{j,t}&=[M_{j,t}-|\Phi(\iota_{j,t})-\delta_j|\cdot k_{j,t}]dt+\sigma_{j,k}dB^j_t+\sigma dW_t,\\
  dY_{j,t}^k&=\bigg[q_t|\Phi(\iota_{j,t})-\delta_j|-\mu_t^qq_t+q_tr\bigg]Y_{j,t}^wdt+|\Phi(\iota_{j,t})-\delta_j|\cdot Y_{j,t}^kdt-\phi_j q_t\sigma_t^qZ_{j,t}dt\\
    & \hspace{1cm} {} 
    +\sum_{i=1}^nZ_t^{Bkji}dW_t^{B,i}+\sum_{i=1}^nZ_t^{kji}dB_t^{i}.
    \end{split}
\end{equation}

Similarly, the FBSDE \eqref{eq: FBSDE9} for the wealth processes derived from the Pontryagin stochastic maximum principle (see e.g Oksendal \& Sulem (2018)) reads:
\begin{equation}  \label{eq: FBSDE9}
   \begin{split}
   dw_{j,t}&=\bigg[rw_{j,t}+q_t(M_{j,t}-|\Phi(\iota_{j,t})-\delta_j|\cdot k_{j,t})+k_{j,t}\mu_t^qq_t+\sigma_t^q\sigma q_t\\
        & \hspace{1cm} {}
        +p'_j\sigma_t^q\sigma_{j,k}q_t-k_{j,t}q_tr-e^{-\frac{\rho t}{\gamma}}(Y_{j,t}^w)^{-\frac{1}{\gamma}}\bigg]dt+q_t\sigma_{j,k}dB_t^j+v_jdW_t,\\
    dY_{j,t}^w&= -rY^w_{j,t}dt+\sum_{i=1}^nZ_t^{Bwji}dW_t^{B,i}+\sum_{i=1}^nZ_t^{wji}dB_t^{i}.
   \end{split}
\end{equation}

It is worth noting that the following backward equations are negligible as they are decoupled from other FBSDEs and independent of the controls:
\begin{equation}
    \begin{split}
       dY_{j,t}^k&=\bigg[q_t|\Phi(\iota_{j,t})-\delta_j|-\mu_t^qq_t+q_tr\bigg]Y_{j,t}^wdt+|\Phi(\iota_{j,t})-\delta_j|\cdot Y_{j,t}^kdt-\phi_j q_t\sigma_t^qZ_{j,t}dt\\
    & \hspace{1cm} {} 
    +\sum_{i=1}^nZ_t^{Bkji}dW_t^{B,i}+\sum_{i=1}^nZ_t^{kji}dB_t^{i}.\\
        Y_{j,T}^k&=0.
    \end{split}
\end{equation}
Also, the other backward equations are not negligible but may be separated from the forward equation as the BSDE is constant with respect to the wealth process:
\begin{equation}
    \begin{split}
     dY_{j,t}^w&= -rY^w_{j,t}dt+\sum_{i=1}^nZ_t^{Bwji}dW_t^{B,i}+\sum_{i=1}^nZ_t^{wji}dB_t^{i}.\\
        Y_{j,T}^w&=-1.
    \end{split}
\end{equation}
Their deterministic solution (i.e. by setting $Z=0$) is achieved by solving the following backward ordinary differential equation: 

\begin{equation}
    \begin{split}
         dY^w_{j,t}&=-rY_{j,t}dt,\\
        Y_{j,T}&=-1.
    \end{split}
\end{equation}

The solution for this type of ODE is elementary and is given explicitly in the form:
\begin{equation}
    \begin{split}
        Y_{j,t}=\exp\bigg\{-\int_t^Trds\bigg\}=e^{-r(T-t)}.
    \end{split}
\end{equation}

\section{Proof of Theorem 4.2}\label{a42}
 First, we solve the fixed point problem by taking the expectation of the forward equations of the capital processes:
\begin{equation}
    \begin{split}
        dm_{j,t}=\bigg[M_{j,t}-|\Phi(\iota_{j,t})-\delta_j|\cdot m_{j,t}\bigg]dt+\sigma_{j,k}dB^j_t+\sigma dW_t.
    \end{split}\label{96}
\end{equation}

 We have:
\begin{equation}
    \begin{split}
        M_{j,t}&=(1-\lambda_j)m_{j,t}+\lambda_j\sum_{i=1}^n\alpha_im_{i,t}\\
        &=(1-\lambda_j+\alpha_j\lambda_j)m_{j,t}+\lambda_j\sum_{i=1,i\neq j}^n\alpha_im_{i,t}.
    \end{split}\label{97}
\end{equation}

Substitute Equations \eqref{97} into Equations \eqref{96}, we have:
\begin{equation}
    \begin{split}
         dm_{j,t}=\bigg[(1-\lambda_j+\lambda_j\alpha_j-|\Phi(\iota_{j,t})-\delta_j|\cdot m_{j,t}+\lambda_j\sum_{i=1,i\neq j}^n\alpha_im_{i,t}\bigg]dt+\sigma_{j,k}dB^j_t+\sigma dW_t.
    \end{split}\label{98}
\end{equation}

We can rewrite system \eqref{98} into a succinct form as follows:
\begin{align}
    dm_t=A_tm_tdt+LdB^l_t+SdW^l_t,\label{126}
\end{align}
where $A_t$ is a $n\times n$ matrix with $A_{ij}=\lambda_i\alpha_j$ for all $i,j$ such that $i\neq j$, and $A_{ii}=1-\lambda_i+\lambda_i\alpha_i-|\Phi(\iota_{i,t})-\delta_i|$. \\

Also, $S$, and $L$ are both $n\times n$ matrix with $S_{ii}=\sigma$, $L_{ii}=\sigma_{i,k}$ for all $i$ and $S_{ij}=0$, $L_{i,j}=0$ for all $i,j$ such that $i\neq j$. And, $B^l_t$, and $W^l_t$ are $n$-dimensional Brownian motions such that $B^l_t=(B^1_t,B^2_t,...,B^n_t)$, and $W^l_t=(W_t,W_t,...,W_t)$.\\

The linear SDE \eqref{126} has a unique analytical solution (see e.g Shreve \& Karatzas (1998)), which is represented in the following form:
\begin{align}
    m_t=v_t\bigg[m_0+\int_{0}^tv^{-1}_sL dB^l_s+\int_{0}^tv^{-1}_sS dW^l_s\bigg]; \quad 0\leq t\leq T,
\end{align}
where $v_t=e^{\int_0^\infty A_sds}$ is the fundamental solution of the following homogeneous equation:
\begin{align*}
    dv_t=A_tv_tdt; \quad v_0=I,
\end{align*}
where $I$ is the $n\times n$ identity matrix.

\section{Proof of Proposition 5.1}\label{a51}
We construct the adjoint FBSDEs and find the corresponding Hamiltonian: 
\begin{equation}
\begin{split}
    H^j=\sum_{s=1}^n\bigg[\bigg\langle y_s^w,\int_Cb^w_s(.)\pi_{c,t}^s(c)dc\bigg\rangle+u_j^k\cdot y_s^k+
    \bigg\langle z_s^{\text{w}},\int_\varphi v_s^w(.)\pi_{\phi,t}^s(\phi)d\phi\bigg\rangle
    \bigg]\\
    +e^{-\rho_jt}\bigg[\int_CU(c_{j,t})\pi_{c,t}^j(c)dc-\int_\varphi\frac{\phi_j^2}{2}\pi_{\phi,t}^j(\phi)d\phi-\lambda_c\int_C\pi_{c,t}^j(c)\ln \pi_{c,t}^j(c)dc\\
    -\lambda_\phi\int_\varphi\pi_{\phi,t}^j(\phi)\ln \pi_{\phi,t}^j(\phi)d\phi\bigg]
\end{split}    
\end{equation}
Minimizing the Hamiltonian for the saving strategy $\pi_{c,t}^j$, we have:
\begin{align}
    \frac{\partial H^j}{\partial \pi_{c,t}^j}=-y_j^w+e^{-\rho_jt}\bigg[(c_{j,t})^{-\gamma}-\lambda_c(\ln \pi_{c,t}^j +1)\bigg]
\end{align}

Hence, the optimal saving strategy is:
\begin{align}
  \widehat{\pi}_{c,t}^j(c)=\exp\bigg[\frac{(c_{j,t})^{-\gamma}-y_j^w\cdot e^{\rho_jt}}{\lambda_c}-1\bigg]  
\end{align}

Minimizing the Hamiltonian for the risk strategy $\pi_{\phi,t}^j$, we have:
\begin{align}
    \frac{\partial H^j}{\partial \pi_{\phi,t}^j}=z_j^wq_t(k_{j,t}\sigma_t^q+\sigma)-e^{-\rho_jt}\bigg[\phi_j+\lambda_\phi(\ln \pi_{\phi,t}^j +1)\bigg]
\end{align}

Hence, the optimal risk strategy is:
\begin{align}
  \widehat{\pi}_{\phi,t}^j(\phi)=\exp\bigg[\frac{z_j^wq_t(k_{j,t}\sigma_t^q+\sigma)\cdot e^{\rho_jt}-\phi_j}{\lambda_\phi}-1\bigg]  
\end{align}

The system of backward equations satisfies:
\begin{equation}
\begin{split}
    dY_{j,t}^w&=-\partial_{w_{j,t}}H^j_{\text{full}}(.)dt+\sum_{i=1}^nZ_t^{Bwji}dW_t^{B,i}+\sum_{i=1}^nZ_t^{wji}dB_t^{i}\\
    &=-Y^w_{j,t}\cdot\bigg(\int_Cr\pi_{c,t}^j(c)dc\bigg)dt+\sum_{i=1}^nZ_t^{Bwji}dW_t^{B,i}+\sum_{i=1}^nZ_t^{wji}dB_t^{i},\\
     &=-rY^w_{j,t}dt+\sum_{i=1}^nZ_t^{Bwji}dW_t^{B,i}+\sum_{i=1}^nZ_t^{wji}dB_t^{i},\\
    dY_{j,t}^k&=-\partial_{k_{j,t}}H^j_{\text{full}}(.)dt+\sum_{i=1}^nZ_t^{Bkji}dW_t^{B,i}+\sum_{i=1}^nZ_t^{kji}dB_t^{i}\\
    &=\bigg[q_t|\Phi(\iota_{j,t})-\delta_j|-\mu_t^qq_t+q_tr\bigg]Y_{j,t}^wdt+|\Phi(\iota_{j,t})-\delta_j|\cdot Y_{j,t}^kdt-\phi_j q_t\sigma_t^qZ_{j,t}dt\\
    & \hspace{1cm} {} 
    +\sum_{i=1}^nZ_t^{Bkji}dW_t^{B,i}+\sum_{i=1}^nZ_t^{kji}dB_t^{i},
\end{split}
\end{equation}
with terminal condition $Y_{j,T}=-1$ and $Y^k_{j,T}=0$ for all $j$. We omit the non-diagonal adjoint processes which can be treated analogously.\\

Similar to previous sections, as we can decoupled the FBSDEs for the capital process, we only consider the FBSDEs for the wealth process here.  The FBSDE \eqref{eq: FBSDE10} for the wealth processes derived from the Pontryagin stochastic maximum principle (see e.g Oksendal \& Sulem (2018)) reads:
\begin{equation}  \label{eq: FBSDE10}
   \begin{split}
   dw_{j,t}&=\int_C\bigg[\bigg(rw_{j,t}+q_t(M_{j,t}-|\Phi(\iota_{j,t})-\delta_j|\cdot k_{j,t})+k_{j,t}\mu_t^qq_t+\sigma_t^q\sigma q_t\\
        & \hspace{1cm} {}
        +p'_j\sigma_t^q\sigma_{j,k}q_t-k_{j,t}q_tr-c_{j,t}\bigg)\exp\bigg[\frac{(c_{j,t})^{-\gamma}-Y_{j,t}^w\cdot e^{\rho_jt}}{\lambda_c}-1\bigg]dc\bigg]dt\\
        & \hspace{1cm} {}
        +q_t\sigma_{j,k}dB_t^j+\sqrt{\int_\varphi v_j\cdot \exp\bigg[\frac{Z_j^wq_t(k_{j,t}\sigma_t^q+\sigma)\cdot e^{\rho_jt}-\phi_j}{\lambda_\phi}-1\bigg]  d\phi}dW_t,\\
    dY_{j,t}^w&= -rY^w_{j,t}dt+\sum_{i=1}^nZ_t^{Bwji}dW_t^{B,i}+\sum_{i=1}^nZ_t^{wji}dB_t^{i}.
   \end{split}
\end{equation}

The deterministic solution of the backward equation (i.e. by setting $Z=0$) is achieved by solving the following backward ordinary differential equation: 

\begin{equation}
    \begin{split}
         dY^w_{j,t}&=-rY_{j,t}dt,\\
        Y_{j,T}&=-1.
    \end{split}
\end{equation}

The solution for this type of ODE is elementary and is given explicitly in the form:
\begin{equation}
    \begin{split}
        Y_{j,t}=\exp\bigg\{-\int_t^Trds\bigg\}=e^{-r(T-t)}.
    \end{split}
\end{equation}
\end{document}